\documentclass{article}
\usepackage{authblk}
\usepackage[utf8]{inputenc}
\usepackage{natbib}
\usepackage{amsthm,amsmath,stmaryrd,bbm,amsfonts,amssymb,mathabx,graphicx}
\usepackage[normalem]{ulem}
\usepackage{color}
\newcommand{\po}{\left(}
\newcommand{\pf}{\right)}
\newcommand{\co}{\left[}
\newcommand{\cf}{\right]}
\newcommand{\cco}{\llbracket}
\newcommand{\ccf}{\rrbracket}
\newcommand{\R}{\mathbb R}

\newcommand{\N}{\mathbb N} 
\newcommand{\dd}{\text{d}}
\newcommand{\na}{\nabla}
\newcommand{\1}{\mathbbm{1}}

\newtheorem{thm}{Theorem}
\newtheorem{prop}{Proposition}
\newtheorem{lem}{Lemma}

\newtheorem{defi}{Definition}
\newtheorem{conj}{Conjecture}
\newtheorem*{rem*}{Remark}

\title{Impacts of Tempo and Mode of Environmental Fluctuations on Population Growth: Slow- and Fast-Limit Approximations of Lyapunov Exponents for Periodic and Random Environments }

\author[1,2]{Pierre Monmarché}
\author[3]{Sebastian J. Schreiber}
\author[4]{Édouard Strickler}
\affil[1]{Sorbonne Université,
4 place Jussieu 75005 Paris (France)}
\affil[2]{Institut universitaire de France (IUF)}
\affil[3]{Department of Evolution and Ecology and Center for Population Biology,  University of California, Davis 95616 USA}
\affil[4]{Université de Lorraine, CNRS, Inria, IECL, F-54000 Nancy (France)}

\begin{document}

\maketitle

\begin{abstract}
Populations consist of individuals living in different states and experiencing temporally varying environmental conditions. Individuals may differ in their geographic location, stage of development (e.g. juvenile versus adult), or physiological state (infected or susceptible). Environmental conditions may vary due to abiotic (e.g. temperature) or biotic (e.g. resource availability) factors. As survival, growth, and reproduction of individuals depend on their state and the environmental conditions, environmental fluctuations often impact population growth. Here, we examine to what extent the tempo and mode (i.e. periodic versus random) of these fluctuations matter for population growth. To this end, we model population growth for a population with $d$  individual states and experiencing $N$ different environmental states. The models are  switching, linear ordinary differential equations $x'(t)=A(\sigma(\omega t))x(t)$ where $x(t)=(x_1(t),\dots,x_d(t))$ corresponds to the population densities in the $d$ individual states, $\sigma(t)$ is a piece-wise constant function representing the fluctuations in the environmental states $1,\dots,N$, $\omega$ is the frequency of the environmental fluctuations, and $A(1),\dots,A(n)$ are Metzler matrices representing the population dynamics in the environmental states $1,\dots,N$. $\sigma(t)$ can either be a periodic function or correspond to a continuous-time Markov chain. Under suitable conditions, there exists a Lyapunov exponent $\Lambda(\omega)$ such that $\lim_{t\to\infty} \frac{1}{t}\log\sum_i x_i(t)=\Lambda(\omega)$ for all non-negative, non-zero initial conditions $x(0)$ (with probability one in the random case). For both  random and periodic switching, we derive analytical first-order and second-order approximations of $\Lambda(\omega)$ in the limits of slow  ($\omega\to 0$) and fast ($\omega\to\infty$) environmental fluctuations. When the order of switching and the average switching times are equal, we show that the first-order approximations of $\Lambda(\omega)$ are equivalent in the slow-switching limit, but not in the fast-switching limit. Hence, the mode (random versus periodic) of switching matters for population growth. We illustrate our results with applications to a simple stage-structured model and a general spatially structured model. When dispersal rates are symmetric, the first order approximations suggest that population growth rates increase with the frequency of switching -- consistent with earlier work on periodic switching. In the absence of dispersal symmetry, we demonstrate that $\Lambda(\omega)$ can be non-monotonic in $\omega$. In conclusion, our results show that population growth rates often depends both on the tempo ($\omega$) and mode (random versus deterministic) of the environmental fluctuations. 
\end{abstract} 

\section{Introduction}\label{sec:intro}
Central to many questions in ecology, evolution, and epidemiology is identifying how environmental conditions determine population growth. This growth depends on the survival, growth, and reproduction of individuals, which in turn depends on the state of an individual ($i$-state). Individuals may differ in their $i$-state based on their geographic location, their stage of development, their genotype, or their behavior~\citep{deangelis-18}. One way to account for these difference in $i$-state on the population state ($p$-state) is using Matrix Population Models (MPM) \citep{caswell-01,deangelis-18}. In these models, there are a finite number $d$ of i-states and the p-state  corresponds to the vector,  $x=(x_1,x_2,\dots, x_d)$, of population densities. For continuous-time MPMs, a matrix $A$ with non-negative off diagonal elements, $a_{ij}$ with $i\neq j$, define the rates at which individuals in state $j$ contribute to individuals in different state $j$. These contributions may correspond to individuals transition to the other state (e.g. dispersing from one geographic location to another, developing from one stage to another, changing their behavior) or producing individuals of the other state. Diagonal elements, $a_{ii}$, of the matrix $A$ can be positive, zero, or non-negative and correspond to the net rate at which individuals leave their current state (including death) and produce more individual of the same state.  Under constant environmental conditions, the continuous-time MPM is a linear system of differential equations
\begin{equation}\label{eq:constant}
x'(t)=Ax(t)
\end{equation}
where $A$, with its non-negative off-diagonal entries, is called a Metzler matrix~\citep{mitkowski-08}.  These continuous-time MPM commonly arise when linearizing a disease-free equilibrium for an epidemiological model, the extinction equilibrium of a demographic model, or a fixation equilibrium of a population genetics model~\citep{burger-00,hethcote-00,kon-etal-04}.  When the matrix $A$ is irreducible, the Perron-Frobenius Theorem implies that  for any non-negative, non-zero initial condition $x(0)$ 
\begin{equation}\label{eq:lyapunov}
\lim_{t\to\infty}\frac{1}{t} \ln  \sum_i x_i(t) =\lambda(A)
\end{equation}
where $\lambda(A)$ is the spectral abscissa of $A$ i.e. the maximum of the real parts of the eigenvalues of $A$. Hence, $\lambda(A)$ determine the long-term growth rate of the population. In particular, when $\lambda(A)$ is positive, the population persists and grows. When $\lambda(A)$ is negative, the population declines exponentially quickly to the extinction state. 

Most populations experience fluctuations in environmental conditions such as temperature, precipitation, and nutrient availability.  As these environmental conditions determine survival, growth, and reproduction of individuals, environmental fluctuations lead to fluctuations in demographic rates.  Understanding the implication of these  fluctuations on population growth is an active topic~\citep{tuljapurkar-90,boyce-etal-06,fay2020can,hilde2020demographic,kortessis2023neglected,paniw2018interactive,schreiber2023partitioning,tuljapurkar2006temporal}. Nearly all of this earlier work, see however~\citep{benaim2019,BLSS23,BLSS24}, is for discrete-time MPMs. One of the simplest approaches to modeling this time-dependency in continuous-time MPMs is to allow $A$ in \eqref{eq:constant} to switch between a finite number of Metzler matrices $A_1,A_2,\dots, A_N$ corresponding to $N$ environmental states  e.g. summer versus winter, wet versus dry period. If $\sigma(t)$ denotes the environmental state at time $t$, then we get a switched, continuous-time MPM
\[
 x'(t) = A_{\sigma( t)} x(t).
\]
If $\sigma(t)$ is a periodic piecewise constant function or given by an irreducible Markov chain, then  there exists the dominant Lyapunov exponent $\Lambda$ such that
\begin{equation}
\label{eq:def_Lyapunov}
    \Lambda  = \lim_{t \to \infty} \frac{1}{t} \ln  \sum_i x_i(t)
\end{equation}
 holds (with probability one in the random case) for any non-negative, non-zero initial condition $x(0)$  (see Lemmas~\ref{lem:lambdaperiod} and \ref{lem:lambdaMarkov} below). As in the constant environment case \eqref{eq:constant}, the Lyapunov exponent  measures the long-term population growth rate of the system. In particular, by a linearization procedure, it allows to determine the persistence or extinction for switched non-linear flows~\citep{benaim2019}.  
 
Environmental fluctuations may occur at daily, yearly, and multi-decadal time scales. This raises the question to what extent does the rate of switching between environmental states influence the population growth rate $\Lambda$?  To tackle this question, consider a fixed choice of the Meltzer matrices $A_1,\dots, A_N$ and environmental trajectory $\sigma(t)$. To manipulate the rate of switching, let $\omega>0$ be the frequency of environmental switching in which case our MPM becomes \begin{equation}
    \label{eq:EDOswitch}
    x'(t) = A_{\sigma(\omega t)} x(t).
\end{equation}
Let $\Lambda(\omega)$ denote the dominant Lyapunov exponent \eqref{eq:def_Lyapunov} associated with the MPM with frequency $\omega$.   When switching is fast ($\omega\rightarrow \infty$), solutions of the switching model \eqref{eq:EDOswitch}  converges to the solution of an ODE $x'=\bar{A} x$ where the averaged matrix $\bar{A}$ corresponds to the appropriate convex combination of the matrices $A_i$. In this limit, $\Lambda(\omega)$ converges to $\lambda(\bar{A})$~\citep{Benaim2014Stability, NoteTopLyapunov, BLSS24, benaim2019,Chitour2021,DU2021313}.  On the contrary, when the switching is slow ( $\omega \to 0$), the Lyapunov exponent converge to a weighted average of the Lyapunov exponents $\lambda(A_i)$ of the matrices $A_i$ \citep{NoteTopLyapunov, BLSS24}. As the Lyapunov exponents of these limits, in general, are different, these earlier results imply that whether populations grow or decline can depend on the frequency of environmental switching. 

These earlier results raise two interesting questions: Do the population growth rates increase or decrease as one approaches these limits? To what extent does the answer depend on whether the environmental switching is periodic or random? Whether periodic switching or random switching between some matrices yields the same behavior is a recurrent question in the study of switched linear systems. For instance, \citet{Chitour2021} provides conditions on the matrices $A_1,\dots,A_N$ under which the maximal Lyapunov exponent obtained by considering the worst deterministic signal $\sigma $ is strictly greater than the maximal probabilistic Lyapunov exponent obtained by considering the worst Markov chain. In particular, this shows that there exists matrices $A_1,\dots,A_N$ such that a suitable deterministic switching signal makes the system \eqref{eq:EDOswitch} unstable (the population persists), while it is stable for all Markovian signals (the population goes extinct). For applications in population dynamics, one is often interested either in the sign of the Lyapunov exponent (i.e. growth or extinction) or in its monotonicity with respect to some key parameters. For specific models, these questions can be answered in the  fast or slow switching regimes with the explicit formula obtained in the present work. In particular, the monotonicity of the population growth in the switching rate is explicitly given in these two  regimes by the sign of the first order term in its expansion.

To address these questions, let $\Lambda_p(\omega)$ and $\Lambda_M(\omega)$ denote the Lyapunov exponents in the periodic case and  the Markovian case, respectively. For the fast  Markovian (random) switching,  \citet{MonmarcheStrickler} answered one of these questions by deriving an asymptotic expansion of $\Lambda_M(\omega)$ in terms of $1/\omega$. Here, we address the remaining questions by deriving an asymptotic expansion of $\Lambda_p(\omega)$ in terms of $1/\omega$, and deriving first order expansions in terms of $\omega$ of $\Lambda_p(\omega)$ and $\Lambda_M(\omega)$ is the slow-switching limit. For slow-switching, we show that the periodic and random approximations of $\Lambda(\omega)$ are equivalent. However, in the limit of fast-switching, these approximation are not equivalent. Indeed, when there is fast switching between two environments, we show the first-order correction  in terms of $\frac{1}{\omega}$ always vanishes and derive a second-order correction in terms of $\frac{1}{\omega^2}$. In contrast, the first order term for Markovian switching is non-zero in general. We present these results in Section~\ref{sec:results} after providing a more detailed description of the periodic and Markovian models in Section~\ref{sec:models}. To illustrate the applicability of our approximations, we apply them to a simple stage-structure model and explore their implications for spatially structure populations in Section~\ref{sec:applications}. Consistent with earlier work of \cite{K22} for periodic environments, our approximations suggest that population growth rates increase with switching rates for spatially structured populations with symmetric dispersal. However, we use these approximations to illustrate how asymmetric dispersal can result in a non-monotonic relationship between switching frequency and population growth rates. Proofs of the main results are in Sections~\ref{sec:proofs_general} and \ref{sec:circular:proof}.

\section{Lyapunov exponents of periodic and random switching models}\label{sec:models}

Here, we provide more details about the switching continuous-time matrix population model in \eqref{eq:EDOswitch} and their Lyapunov exponents $\Lambda(\omega)$ in \eqref{eq:def_Lyapunov}. For these models, $N \ge 1$ corresponds to the number of environmental states, $\cco 1,N\ccf =\{1,2,\dots,N\}$ is the set of environmental states, and the demographic matrices $A_1,A_2,\dots, A_N$ are $d\times d$ Metzler matrices i.e. have non-negative off-diagonal entries. We consider two approaches for specifying the piece-wise constant function $\sigma(t)$ that takes values $\cco 1,N\ccf$. One approach is to assume that $\sigma(t)$ is a periodic, piecewise constant, right-continuous function. The other approach is to assume that $\sigma(t)$ corresponds to a continuous-time Markov chain on $\cco 1,N\ccf$. These approaches are described in detail below. 

The state space for the matrix population model is the non-negative cone  $\mathbb{R}^d_+=[0,\infty)^d$ of $\mathbb{R}^d$. If $A_i$ is Metzler for all $i$ in the model~\eqref{eq:EDOswitch}, then $x(0) \in \mathbb{R}^d_+$ implies that $x(t) \in \mathbb{R}^d_+$ for all $t \geq 0$ i.e. $\mathbb{R}^d_+$ is forward invariant for the dynamics. 
 
For each form of the model, we also present results ensuring the existence of dominant Lyapunov exponent $\Lambda(\omega)$. To do this for the periodic case,  let $r(A)$ denote the spectral radius of $A$ (i.e. the eigenvalue  of $A$ with the largest absolute value) and $\lambda(A)$ denote the spectral abscissa of $A$ (i.e. the eigenvalue of $A$ with largest real part). For a  Metzler and irreducible matrix $A$, we can apply Perron-Frobenius Theorem to $A + r I$ for some large enough $r$, to conclude that $\lambda(A)$ is a simple eigenvalue of $A$. Associated with $A$ is a unique positive right eigenvector $x$ with $\sum_ix_i = 1$ and a unique left eigenvector $y$ with $\sum x_iy_i = 1$. 

\subsection{The periodic case}\label{subsec:periodic}
For the periodic version of the model~\eqref{eq:def_Lyapunov}, we assume that the function $\sigma:\mathbb{R}\to \cco 1, N\ccf$ is piece-wise constant, right-continuous, has period $1$ and  successively takes the values $1, \ldots, N$ as follows
\[
\sigma(t) =  i \mbox{ for all } t \in [\tau_i, \tau_{i+1}),
\]
where $\tau_0 = 0$ and $\tau_{i+1} = \tau_i + \alpha_i$ for some $\alpha_i\ge 0$ satisfying $\sum_{i=1}^N \alpha_i = 1$.  $\alpha_i$ represents the proportion of time spent in environmental state $i$ during  one period of length $T:= \omega^{-1}$. 

Let $M(\omega )$ be the \emph{monodromy matrix} associated with model over the time interval $[0,1]$:
\[
M(\omega) = e^{\omega^{-1} \alpha_N  A_N} \cdots e^{\omega^{-1} \alpha_1  A_1}
\]
where $e^M$ denotes the matrix exponential of the matrix $M$. For all initial condition $x(0)$ of \eqref{eq:EDOswitch}, one can verify that $x(\omega^{-1}) = M(\omega )x(0)$. $M(\omega)$ is a  matrix with nonnegative entries, so that its principal eigenvalue $\lambda(M(\omega))$ is well-defined. We have the following (see e.g. \citep[Thm. 2]{BLSS24DIGORDID})
\begin{lem}
\label{lem:lambdaperiod}
If the matrices $A_1,\dots,A_N$ are Metzler and $\sum_{i=1}^N A_i$ is irreducible, then  \eqref{eq:def_Lyapunov} holds for all non-zero initial conditions $x(0) \in \mathbb{R}_+^d$ and equals
\[
\Lambda_p(\omega) =\omega \log \left( \lambda( M(\omega )) \right).
\]
\end{lem}

\subsection{The Markovian case}\label{subsec:Markov}

For random environmental switching, let $\sigma(t)$ be a continuous-time Markov chain with transition rates $\xi_{ij}\ge 0$ from state $i$ to state $j\neq i$. Let $\Xi$ be the transition matrix with entries $\xi_{ij}$ for $i\neq j$ and $\xi_{ii}=-\sum_{j\neq i} \xi_{ij}$. Assume  $\Xi$ is irreducible in which case there is a unique invariant probability measure $\alpha=(\alpha_1,\dots,\alpha_N)$  on the environmental state space $\cco 1,N \ccf$. Consider  $(x(t),\sigma(\omega t))_{t\geqslant 0}$ where $x$ solves \eqref{eq:EDOswitch}.   $(x(t),\sigma(\omega t))_{t\geqslant 0}$ is known as a Piecewise Deterministic Markov Process (PDMP). 

For a non-negative, non-zero initial condition $x(0)\in\R_+^d$, decompose $x(t)$ as $x(t) = \rho(t)\theta(t)$ where $\rho(t) = \1\cdot x(t)>0$ with  $\1=(1,\dots,1) \in \R^d$ and $\theta(t) = x(t)/\rho(t) \in \Delta = \{x\in\R_+^d, x_1+\dots+x_d=1\}$. In this coordinate system, the switching model \eqref{eq:EDOswitch} is 
\begin{equation}
    \label{eq:rhotheta}
\rho'(t) = (\1\cdot A_{ \sigma(\omega t )}  \theta(t)) \rho(t)\,,\qquad \theta'(t) = F_{\sigma(\omega t)}(\theta(t))
\end{equation}
with
\begin{equation}
    \label{eq:Fi}
    F_i(\theta) = A_i \theta - (\1\cdot A_i\theta )\theta\,.
\end{equation}
The following is proven in \cite{NoteTopLyapunov}.

\begin{lem}
\label{lem:lambdaMarkov}
Assuming that the matrices $A_1,\dots,A_N$ are Metzler and $\sum_{i=1}^N A_i$ is irreducible, then 
the Markov process $(\theta(t),\sigma(\omega t))_{t\geqslant 0}$ admits a unique invariant measure $\mu_\omega$ on $ \Delta \times \cco 1,N\ccf $ and, moreover, the limit \eqref{eq:def_Lyapunov} exists almost surely, is deterministic, independent from the initial condition $x(0) \in \mathbb{R}_+^d\setminus\{0\}$ and is given by
\begin{equation}\label{eq:LyapunovSlowMarkov}
     \Lambda_M(\omega) =   \int_{\Delta \times \cco 1,N\ccf} \1 \cdot A_i \theta\, \mu_\omega ( \dd \theta \dd i)\,.
\end{equation} 
\end{lem}
 
\section{General results}\label{sec:results}

Our most general results are first-order expansions of $\Lambda_p(\omega),\Lambda_M(\omega)$ in terms of $\omega$ in the slow switching limit ($\omega\to 0$) and in terms of $\frac{1}{\omega}$ in the fast switching limit ($\omega\to\infty)$. In the case of only two environments ($N=2$), we show that the first-order correction term of $\Lambda_p(\omega)$ is zero in the fast switching limit, and derive a second order approximation in terms of $\frac{1}{\omega^2}$. We also provide a simpler representation of the first-order approximation of $\Lambda_M(\omega)$ in the random case with two environments. 

\subsection{First order expansions of $\Lambda_p(\omega),\Lambda_M(\omega)$ for fast and slow switching.}\label{subsec:general}

For the slow limit ($\omega\to 0$), we need the right eigenvectors $x_i$ and left eigenvectors $y_i$ of $A_i$ associated with the eigenvalue $\lambda(A_i)$. We assume that these eigenvectors are normalized such that  $x_i^\intercal y_i =1$ and $\1_d^\intercal x_i=1$ where  $\1_d$  is the column vectors  with all entries equal to $1$. For the fast limit ($\omega\to \infty$), define the averaged matrix 
\[\bar{A} = \sum_{i=1}^N \alpha_i A_i\]
with weights $\alpha=(\alpha_1,\dots,\alpha_N)$ given either by $\alpha_i = \tau_{i+1}-\tau_i$ in the periodic settings of Section~\ref{subsec:periodic}, or by the invariant measure of $\sigma$ in the Markovian settings of Section~\ref{subsec:Markov}. Let  $\bar{x}$ and $\bar{y}$   the right and left eigenvectors of $\bar{A}$ associated with eigenvalue $\lambda(\bar{A})$ such that $\bar{x}^\intercal \bar{y} =1$ and $\1_d^\intercal \bar x=1$.

To state our main results, recall that  the  \emph{matrix commutation}  of two $d\times d$ matrix $A$ and $B$ is 
\[
[A,B]= AB - BA.
\]
Furthermore, we need a generalized inverse for non-invertible matrices.  
\begin{defi}
Let $M$ be a matrix of index $1$, meaning that $M$ and $M^2$ have the same rank. Then, the \emph{group inverse} of $M$ is the unique solution $X$ to the matrix equation 
\[
MXM= M, \quad XMX=X, \quad XM = MX.   
\]
When it is well-defined, we denote the group inverse of $M$ by $M^{-1}$.
\end{defi}
\noindent Importantly for us, if $A$ is an irreducible, Metzler matrix, then  $M=A - \lambda(A) I$  is an index $1$, singular matrix.
One way of computing $M^{-1}$ is to use its rank factorization. Specifically if $M$ has rank $k<d$, then there is a $d\times k$ matrix $C$ with rank $k$ and a $k\times d$ matrix $F$ with rank $k$ such that $M=CF$. Then $M^{-1}=C(FC)^{-2}F$. There are multiple ways to get the rank factorization. For example, if $M=U\begin{pmatrix}\Sigma_k & 0 \\ 0 & 0 \end{pmatrix}  V$ is a singular value decomposition of $M$ where  $\Sigma_k$ is a $k\times k$ diagonal matrix with the non-zero singular values of $M$, then one can choose $C$ to be the first $k$ columns of $U$ and $F$ to be the matrix product of $\Sigma_k$ and the first $k$ rows of $V$.

The following theorem provides first order expansions of the Lyapunov exponents $\Lambda_p,\Lambda_M$ in the slow and fast limits. 
\begin{thm}\label{thm:general}
Let $A_1,\dots,A_N$ be Metzler matrices, such that  $\bar{A}$ is irreducible. 
\begin{enumerate}
    \item (Slow periodic switching) If $A_1,\dots,A_N$ are irreducible, then
        \begin{equation}
        \label{eq:thm_general_slow_periodic}
        \Lambda_{p}(\omega) = \sum_{i=1}^N \alpha_i \lambda(A_i) + c_{sp} \omega + \underset{\omega\rightarrow 0}o (\omega)
    \end{equation}
    with 
    \[c_{sp} = \sum_{i=1}^N \ln  \left( x_i^\intercal y_{i+1} \right) \,. \]
     \item (Fast periodic switching) 
             \begin{equation}
        \label{eq:thm_general_fast_periodic}
    \Lambda_{p}(\omega) =  \lambda\po \bar{A}\pf + c_{fp} \omega^{-1} + \underset{\omega\rightarrow \infty }o (\omega^{-1})    
    \end{equation}
        with
    \[c_{fp} =  \bar y^\intercal \left( \frac{1}{2} \sum_{1 \leq i < j \leq N} \alpha_j \alpha_i [A_j, A_i] \right) \bar x  \,. \]
    \end{enumerate}
\noindent For the random switching case, assume, without loss of generality that, $\max_j \sum_{i\neq j}\xi_{i,j}=1$, and define $Q$ to be the jump matrix given by $q_{i,j}= \xi_{i,j}$ for $i\neq j$ and $q_{i,i}=1-\sum_{j\neq i} q_{i,j}$. Then
    \begin{enumerate} \setcounter{enumi}{2}

     \item (Slow Markovian switching) If $A_1,\dots,A_N$ are irreducible, then
            \begin{equation}
        \label{eq:thm_general_slow_Markov}
    \Lambda_{M}(\omega) =   \sum_{i=1}^N \alpha_i \lambda(A_i) + c_{sM} \omega + \underset{\omega\rightarrow 0}o (\omega)     
    \end{equation}
        with
    \[c_{sM} = \sum_{i,j=1}^N \alpha_i q_{ij} \ln (y_j^\intercal x_i) \,.\]
     \item (Fast Markovian switching) 
             \begin{equation}
        \label{eq:thm_general_fast_Markov}
    \Lambda_{M}(\omega) =  \lambda\po \bar{A}\pf + c_{fM} \omega^{-1} + \underset{\omega\rightarrow \infty }o (\omega^{-1})    
    \end{equation}
        with
    \[c_{fM} =   \bar y^\intercal \left( \sum_{i,j} \alpha_i (Q - I)_{i,j}^{-1} A_j ( \bar x \bar y^\intercal - I) A_i \right) \bar x\,. \]
\end{enumerate}
\end{thm}
\noindent The fast Markovian case \eqref{eq:thm_general_fast_Markov} follows from \citep[Proposition 2]{MonmarcheStrickler}. The proofs for the three other cases are given in Section~\ref{sec:proofs_general}.

The signs of the first order terms determine whether the Lyapunov exponents are increasing or decreasing in the slow and fast switching limits. For instance, in the slow switching limit, a positive value of $c_{sp} >0$ (resp. $c_{Mp}$) implies that there there exists a positive frequency $\omega_0>0$ such that $\Lambda_p$ (resp. $\Lambda_M$) is increasing  over the interval $[0,\omega_0]$. Alternatively, in the fast switching limit, a positive value of $c_{fp} >0$ (resp. $c_{fp}$) implies that there there exists a positive frequency $\omega_\infty>0$ such that $\Lambda_p$ (resp. $\Lambda_M$) is decreasing  over the interval $[\omega_\infty,\infty)$. 

\paragraph{Correspondence between periodic and random cases.} For the slow switching case, there is a strong correspondence between the periodic and random approximations, as we now show. Let $\alpha_i$ be given for the periodic version of the model (see section~\ref{subsec:periodic}). The natural, random counterpart of this deterministic model  follows the environmental states in the same order and, on average, remains in each environmental state for the same amount of time. To this end, let $\beta=\min_i \alpha_i$ and define the transition matrix $\Xi$ by $\xi_{i,i+1}= \beta/\alpha_i$ for $1\le i \le N-1$, $\xi_{N,1}=\beta/\alpha_N$ and all other off-diagonal coefficients are $0$. This transition matrix is such that the Markov process $\sigma$ visits successively the states $1, \ldots, N$,  its unique invariant probability measure is $\alpha = ( \alpha_1, \ldots, \alpha_N)$, and $\max_i \sum_{j\neq i} \xi_{i,j}=1$. The average time for $\sigma$ to visit the $N$ successive states (corresponding to one period of the deterministic signal) is $1/\beta$. Therefore, $\Lambda_M(\omega)$ should be compared with $\Lambda_p(\beta \omega)$. The asymptotic expansion \eqref{eq:thm_general_slow_Markov} for the slow switching rate reads
\begin{equation}\label{eq:comparison}
  \begin{aligned}
    \Lambda_{M}(\omega ) & =   \sum_{i=1}^N \alpha_i \lambda(A_i) +  \omega  \sum_{i,j} \alpha_i q_{i,j} \ln( y_j^\intercal x_i) + \underset{\omega\rightarrow 0}o (\omega)    \\ 
    & =   \sum_{i=1}^N \alpha_i \lambda(A_i) +  \omega \beta \sum_{i=1}^N  \ln( y_j^\intercal x_i) + \underset{\omega\rightarrow 0}o (\omega)\,,
    \end{aligned}
    \end{equation}
which is exactly the formula of $\Lambda_p( \beta \omega)$ in the periodic case \eqref{eq:thm_general_slow_periodic}.  This correspondence doesn't occur in the fast switching regime. Indeed, even when $N=2$, Propositions~\ref{prop:N2periodic} and \ref{prop:N2Markov} in the next section provide examples where the periodic and Markovian cases differ at first order.

\paragraph{The problem with reducible $A_i$ in the slow limit.} 
In Theorem \ref{thm:general}, the slow switching formulae are proven under the additional assumption that each $A_i$ is irreducible. To see why this assumption is necessary, we  provide an example involving three, reducible $2\times 2$ matrices such that $\lim_{\omega \to 0} \Lambda_p(\omega) \neq \sum_i \alpha_i \lambda(A_i)$. Namely,  the expansion \eqref{eq:thm_general_slow_periodic} can fail at zero$^{\mbox{th}}$ order for reducible matrices. Consider the three, reducible matrices
\[
A_1 = \begin{pmatrix}
2 & 1\\
0 & 1
\end{pmatrix}, \quad A_2 = \begin{pmatrix}
1 & 0\\
0 & 2
\end{pmatrix}, \mbox{ and } A_3 = \begin{pmatrix}
2& 0\\
1 & 1
\end{pmatrix}
\]
with $\lambda(A_1)=\lambda(A_2) = \lambda(A_3) = 2$.  Let $\alpha_1=\alpha_2=\alpha_3=1/3$. The matrix $\bar A=\frac{1}{3}\sum_i A_i $ is irreducible. Writing  $T = \omega^{-1}$, one can compute 
\[
\prod_{i=1}^3 e^{A_iT/3}=e^{T}\begin{pmatrix}e^{2 T/3} & \left(e^{T/3} - 1\right) e^{ T/3}\\\left(e^{T/3} - 1\right) e^{ T/3} & \left(e^{T/3} - 1\right)^{2} + e^{T/3} \end{pmatrix}
\]
whose dominant eigenvalue is 
\[
\frac{\sqrt{\left(2 e^{\frac{5 T}{3}} - e^{\frac{4 T}{3}} + e^{T}\right)^{2} - 4 e^{3 T}}}{2} + e^{\frac{5 T}{3}} - \frac{e^{\frac{4 T}{3}}}{2} + \frac{e^{T}}{2} = 2 e^{5T/3}+O(e^{-T/3}).
\]
It follows that  
\begin{equation}
  \label{eq:notthegoodlimit}  
\lim_{\omega \to 0} \Lambda_p(\omega)=\frac{5}{3} <2= \frac{1}{3}\sum_i  \lambda(A_i).
\end{equation}

\subsection{Fast switching with two matrices}\label{subsec:proof_fast_Markov}

In this section, we investigate the case  when switching only between two matrices. In the periodic case, we show that the first order term of the expansion is always $0$, and give a second order expansion.

\begin{prop}\label{prop:N2periodic}
In the fast periodic case of Theorem~\ref{thm:general}, consider the case $N=2$.
Then, $c_{fp} = 0$. Moreover, for any Metzler matrices $A_1, A_2$ of size $d \times d$,
\begin{equation}
    \label{eq:expansion2}
    \Lambda_p(\omega) = \lambda(\bar A) + c_{fp,2} \omega^{-2} +  \underset{\omega\rightarrow \infty }o (\omega^{-2})    
\end{equation}
with
\[
c_{fp,2} = \bar y^\intercal \Big( \frac{1}{12} \left[\alpha_2 A_2-\alpha_1 A_1,[\alpha_2 A_2, \alpha_1 A_1]\right] +\frac{1}{4}\alpha_1^2 \alpha_2^2 [A_2,A_1]\left(\bar A - \lambda(\bar A) I\right)^{-1}[A_2,A_1] \Big) \bar x   \,. 
\]
\end{prop}

\begin{proof}
 First, we show that the first order term is zero. By \eqref{eq:thm_general_fast_periodic}, this term is given by $c_{fp}=\frac{1}{2} \bar y^\intercal [\alpha_2 A_2,\alpha_1 A_1]\bar x$. Set $\bar \lambda = \lambda(\bar A)$. By definition of $\bar x$ and $\bar y$, one has 
\[\alpha_2 A_2 \bar x = ( \bar \lambda - \alpha_1 A_1) \bar x \quad \mbox{and} \quad \bar y^\intercal \alpha_2 A_2  = \bar y^\intercal( \bar \lambda - \alpha_1 A_1).
\]
Thus, 
\begin{align*}
\bar y^\intercal[\alpha_2 A_2,\alpha_1 A_1]\bar x &= \bar y^\intercal \alpha_2 A_2 \alpha_1 A_1 \bar x - \bar y^\intercal \alpha_1 A_1 \alpha_2 A_2 \bar x\\
&=\bar y^\intercal ( \bar \lambda - \alpha_1 A_1) \alpha_1 A_1 \bar x - \bar y^\intercal \alpha_1 A_1( \bar \lambda - \alpha_1 A_1)\bar x =0,
\end{align*}
and $c_{fp}=0$. Using once again the Baker- Campbell - Hausdorff formula, we have $M(T) = \exp(T B(T))$, with
\[
B(T ) =    \bar{A} + \frac{T }{2}[\alpha_2 A_2, \alpha_1 A_1] + \frac{T ^2}{12}\left[\alpha_2 A_2 - \alpha_1 A_1, [\alpha_2 A_2, \alpha_1 A_1] \right] + o(T ^2).
\]
By Theorem 4.1 in \cite{HRR92}, we have for a matrix $F$ and $T $ small enough (recall that $\bar \lambda I - \bar A$ admits a group inverse)
\[
\lambda( \bar{A} + T  F) = \lambda(\bar{A}) + T  \bar y^\intercal F \bar x + T^2 \bar y^\intercal F ( \bar \lambda I - \bar A)^{-1} F \bar x +o(T ^2).
\]
Thus, 
\begin{multline*}
    \Lambda_p(T^{-1}) = \bar \lambda + T  \bar y^\intercal \left( \frac{1 }{2}[\alpha_2 A_2, \alpha_1 A_1]+ \frac{T}{12}\left[\alpha_2 A_2 - \alpha_1 A_1, [\alpha_2 A_2, \alpha_1 A_1]\right] \right)  \bar  x \\
    + T ^2  \bar y^\intercal \frac{1}{4}\alpha_1^2 \alpha_2^2 [A_2, A_1](  \bar  \lambda I - \bar A)^{-1}[A_2,A_1]  \bar  x + o(T ^2),
\end{multline*}
which gives the announced result. 

\end{proof}

We now provide a simpler formula for the first order term in \eqref{eq:thm_general_fast_Markov} in the Markov switching case. In this case, the transition matrix is of the form
\begin{equation}
    \label{eq:Qp1-p}
    Q = \begin{pmatrix}
1-p & p\\
  1-p & p
\end{pmatrix},
\end{equation}
where $p=\xi_{1,2}\in(0,1)$.

\begin{prop}\label{prop:N2Markov}
In the fast Markov case of Theorem~\ref{thm:general}, consider the case $N=2$, with a transition matrix  given by \eqref{eq:Qp1-p}.  Then, for any Metzler matrices $A_1, A_2$ of size $d \times d$, \eqref{eq:thm_general_fast_Markov} holds with
\begin{equation}
    \label{eq:c12matrices}
    c_{fM} = p(1-p) \left[  \bar y^\intercal(A_1 - A_2)^2 \bar x -  \left(\bar y^\intercal (A_1 - A_2) \bar x\right)^2 \right].
\end{equation}
In particular, if $A_1 - A_2$ is a diagonal matrix with entries $(\gamma_i)_{1 \leq i \leq d}$ on the diagonal, then
\[
c_{fM} =  (1-p) p\co \sum_i \gamma_i^2 \bar x_i \bar y_i  -  \left( \sum_i \gamma_i \bar x_i \bar y_i \right)^2\cf  \geq 0,
\]
with strict inequality as soon as the $\gamma_i$ are not all the same.
\end{prop}

In particular, when switching randomly between two matrices, the first order term in the fast regime expansion is, in general, not zero, contrary to the periodic switching case.
\begin{proof}
The invariant measure of $Q$ is $\alpha = (1-p, p)$, and it is easily checked that  $(Q-I)^{-1} =  (Q-I)$. Hence, \eqref{eq:thm_general_fast_Markov}  reads
\[
c_{fM} = (1-p)p \bar y^\intercal \left( A_1P A_2 + A_2 P A_1 - A_1 P A_1 - A_2 P A_2 \right) \bar x,
\]
with $P = I - \bar x \bar y^\intercal$. Straightforward computations lead to Equation \eqref{eq:c12matrices}. Finally, since $\sum_i \bar x_i \bar y_i = 1$, the fact that $c_1 \geq 0$ in the case where $A_1 - A_2$ is diagonal follows from Jensen's inequality.
\end{proof}

\section{Applications}\label{sec:applications}

\subsection{Simple stage-structured models}

To illustrate the use of these approximations, we examine a stage-structured model with juveniles ($x_1$) and adults ($x_2$). Juveniles  mature to adults at rate $1$. We consider two versions of the model: one with fluctuating birth rates and one with fluctuating mortality rates. For both versions, there are two environments and  the  environment transitions between its states at rate $1$ when $\omega=1$. Hence, the transition probabilities of the jump process is $  Q=\begin{pmatrix} 0.5& 0.5\\ 0.5 & 0.5 \end{pmatrix}$.

For the model with fluctuating birth rates, adults die at rate $1$, reproduce at rate $a>0$ in environment $1$, and reproduce at rate $b>0$ in environment $2$. Assume that $a>b$ i.e. the birth rate is higher in environment $1$. Hence, 
\begin{equation}\label{eq:JA1}
A_1=\begin{pmatrix} -1& a\\ 1& -1 \end{pmatrix} \mbox{ and }
 A_2=\begin{pmatrix} -1& b\\ 1 & -1 \end{pmatrix} .
\end{equation}
In the slow switching limit $(\omega\to 0)$, we get the approximation 
\[
\Lambda(\omega) =\bar{\lambda}_s+ c_s \omega + o(\omega^2)
\]
where
\[
\bar{\lambda}_s=\frac{\sqrt{a}}{2} + \frac{\sqrt{b}}{2} - 1 \mbox{ and }
c_s=
\log{\left( \frac{\left(\sqrt{a} + \sqrt{b}\right)^{2}}{4 \sqrt{ab}}
 \right)}.
\]
The zero-th order term $\bar{\lambda}_s$ corresponds to the average $\frac{1}{2}\sum_i \lambda(A_i)$ and is positive if $\sqrt{a}+\sqrt{b}>2$. In particular, one needs $a>1$ to ensure that individuals have a chance of replacing themselves in their lifetime. The second order term equals zero if $a=b$ and is positive otherwise. Hence, provided there is variation in birth rates, the population growth rate increases with frequency $\omega$  (at low frequencies). 

In fast randomly switching environments ($\omega\to\infty$), 
\[
\Lambda_{M}(\omega)=\lambda(\bar A)+ c_{fM}\frac{1}{ \omega}+ o\left(\frac{1}{\omega^2}\right)
\]
The zero-th order term equals
\[
\lambda(\bar A)=\frac{\sqrt{2( a +  b)}}{2} - 1
\]
and is positive only if the average birth rate $\frac{a+b}{2}$ is greater than one. Moreover, as we have $\lambda(\bar A)>\bar{\lambda}_s$, the long-term population growth rate is higher in the fast limit $\omega\to \infty$ than in the slow limit $\omega \to 0$. The first order correction term is 
\[c_{fM}=
- \frac{\left(a - b\right)^{2}}{8 a + 8 b}.
\]
This term is negative whenever there is variation in the birth rates. In the fast periodically switching environments ($\omega\to\infty$), 
\[
\Lambda_{p}(\omega)=\lambda(\bar A)+ c_{fp,2}\frac{1}{ \omega^2}+ o\left(\frac{1}{\omega^3}\right)
\]
where the second-order term equals
\[
c_{fp,2}=
- \frac{7 \sqrt{2} \left(a - b\right)^{2}}{384 \sqrt{a + b}}.
\]
This term is also negative whenever there is variation in the birth rates. 

Taken together these approximations suggest that (i) the long-term population growth rate $\Lambda(\omega)$ increases with frequency $\omega$ and (ii) the long-term population growth rate is higher in the periodic environment than the random environment. To explore what happens at intermediate frequencies, we numerically calculated $\Lambda_p(\omega)$ for $a=2.5$ (good times) and $b=0.01$ (bad times) (Fig.~\ref{fig:AJ}A). For these birth rates, $\Lambda(0)<0$ and $\lim_{\omega\to\infty}\Lambda(\omega)>0$. As predicted by the approximations, the long-term population growth rate is increasing with $\omega$ and is always lower for the random environment. Interestingly, at intermediate frequencies (around $\omega\approx 0.5$), the long-term population growth rate is positive for a periodic environment but negative for the random environment. Hence at these intermediate frequencies, populations persist in periodically switching environments but not randomly switching environments.

For the model with fluctuating death rates, adults have a per-capita birth rate of $a$. In environment $1$, adults have a per-capita death rate of $b$ and juveniles have a per-capita death rate of $1$. In environment $2$, juveniles have a per-capita death rate $b$, while adults have a per-capita death rate $1$.  Hence, 
\begin{equation}\label{eq:JA2}
A_1=\begin{pmatrix} -b& a\\ 1& -1 \end{pmatrix} \mbox{ and }
A_2=\begin{pmatrix} -1& a\\ 1 & -b \end{pmatrix} 
\end{equation}
We assume $b>1$ i.e. environment $1$ is worse for adults, while environment $2$ is worse for juveniles.

\begin{figure}
\center\includegraphics[width=0.45\textwidth]{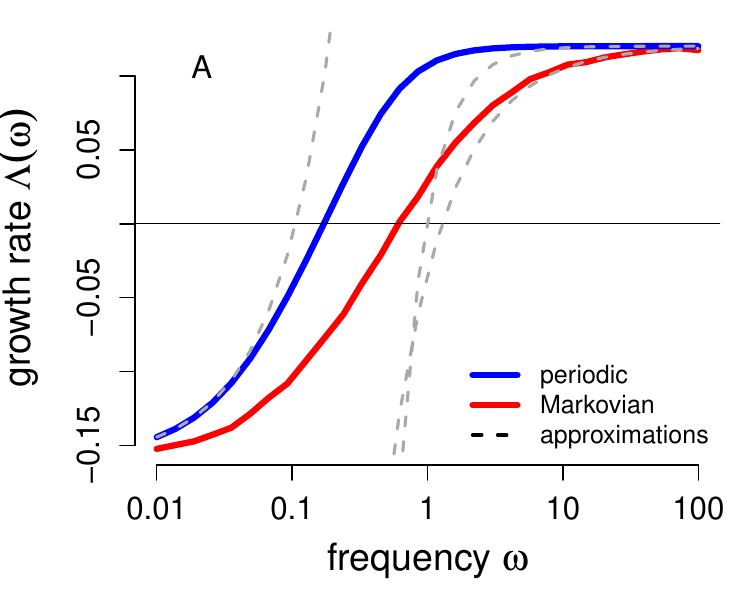}\includegraphics[width=0.45\textwidth]{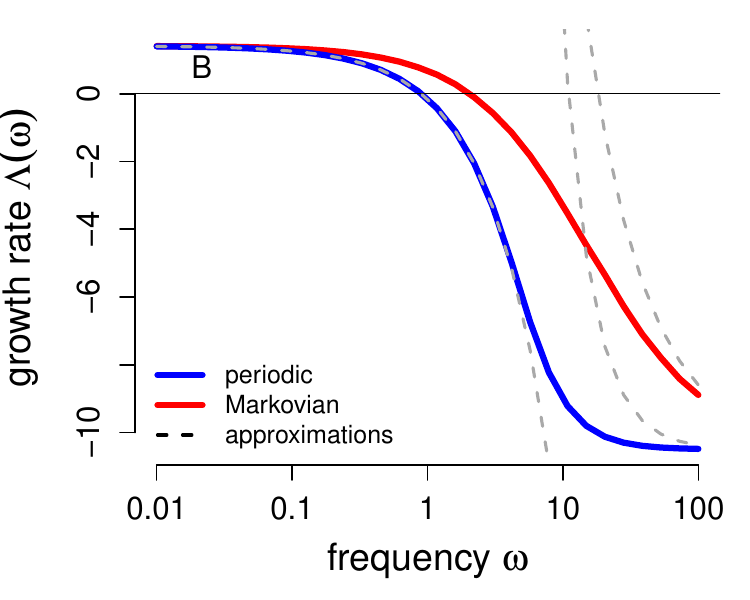}
\caption{In A and B, the Lyapunov exponents for the juvenile-adult models~\eqref{eq:JA1} and \eqref{eq:JA2}, respectively, as functions of the switching frequency $\omega$. Solid lines correspond to numerically approximated Lyapunov exponents for periodic (blue) and random (red) switching. The dark grey dashed lines correspond to the analytical approximations. Parameters: In A, per-capita birth rates are $a=2.5$ in environment $1$ and $b=0.01$  in environment $2$. In B, the per-capita birth rate is $a=100$ in both environments, the per-capita death rates for juveniles and adults are $b=40$ and $1$, respectively, in environment $1$, and the per-capita death rates are $1$ and $b=40$, respectively, in environment $2$. }\label{fig:AJ}
\end{figure}

In the slow switching limit $(\omega\to 0)$, we get the approximation 
\begin{equation*}
\Lambda(\omega) =\bar{\lambda}_s+ c_s \omega + o(\omega^2)
\end{equation*}
where
\[
\bar{\lambda}_s=- \frac{b}{2} + \frac{\sqrt{4 a + b^{2} - 2 b + 1}}{2} - \frac{1}{2}
\mbox{ and }
c_s=\log \left( \frac{4 a}{4 a + (b-1)^2}
\right).
\]
The zero-th order term $\bar{\lambda}_s$ corresponds to the average $\frac{1}{2}\sum_i \lambda(A_i)$ and is positive if the birth rate $a$ is greater than the maximal per-capita death rate $b>1$. Under our assumption $b>1$, the first-order term is negative. Hence, at lower frequencies, the population growth rate decreases with frequency. 

Alternatively, in fast randomly switching environments ($\omega\to\infty$), 
\begin{equation*}
\Lambda_{M}(\omega)=\lambda(\bar A)+ c_{fM}\frac{1}{ \omega}+ o\left(\frac{1}{\omega^2}\right)
\end{equation*}
The zero-th order term 
\[
\lambda(\bar A)=\sqrt{a} - \frac{b}{2} - \frac{1}{2}
\]
 is positive only if the square root of the average birth rate $\sqrt{a}$ is greater than the average death rate $\frac{b+1}{2}$. Moreover, as we have $\lambda(\bar A)>\bar{\lambda}_s$, the long-term population growth rate is higher in the fast limit $\omega\to \infty$ than the slow limit $\omega \to 0$. The first order correction term for the fast, randomly switching environment  
\[c_{fM}=
\frac{\left(b - 1\right)^{2}}{4}
\]
 is always positive. In fast, periodically switching environments ($\omega\to\infty$), 
\begin{equation}
\Lambda_{p}(\omega)=\lambda(\bar A)+ c_{fp,2}\frac{1}{ \omega^2}+ o\left(\frac{1}{\omega^3}\right)
\end{equation}
where the second-order term 
\[
c_{fp,2}=
\frac{7 \sqrt{a} \left(b - 1\right)^{2}}{96}
\]
is also positive. 

Taken together these approximations suggest that (i) the long-term population growth rate $\Lambda(\omega)$ decreases with frequency $\omega$ and (ii) the long-term population growth rate is lower in the periodic environment than the random environment. To explore what happens at intermediate frequencies, we numerically calculated $\Lambda_p(\omega)$ for $a=100$ and $b=40$ (Fig.~\ref{fig:AJ}B). For these birth and death rates, $\Lambda(0)>0$ and $\lim_{\omega\to\infty}\Lambda(\omega)<0$. As predicted by the approximations, the long-term population growth rate is decreasing with $\omega$ and is always higher for the random environment. Hence, at intermediate frequencies (around $\omega\approx 2$), the long-term population growth rate is positive for a random environment but negative for the periodic environment.

\subsection{Metapopulation models}

To model the effects of spatial heterogeneity and environmental fluctuations on population growth, we consider a metapopulation occupying $d$ distinct patches i.e. a population of $d$ populations coupled by dispersal~\citep{hanski-99}. The population density in patch $i$ is $x_i$ with a per-capita growth rate of $r^i_{\sigma(\omega t)}$ at time $t$. Individuals from patch $j$ disperse to patch $i$ at rate $l_{i,j}$. The population dynamics in patch $i$ are 
\begin{equation}
    \label{eq:DIG}
    x_i'(t) = r^i_{\sigma(\omega t)} + \sum_{j \neq i} l_{i,j} ( x_j - x_i).
\end{equation}
In matrix form, the metapopulation dynamics are
\begin{equation}
    \label{eq:DIG-matrix}
    x'(t) = \underbrace{\left( R_{\sigma(\omega t)} + L\right)}_{=:A_{\sigma(\omega t)}}x(t).
\end{equation}
where $R_{\sigma}$ is a diagonal matrix with entries $r_\sigma^1,\dots,r_\sigma^d$ and $L$ is the dispersal matrix with entries $l_{i,j}$. 

\begin{figure}
\center\includegraphics[height=3in]{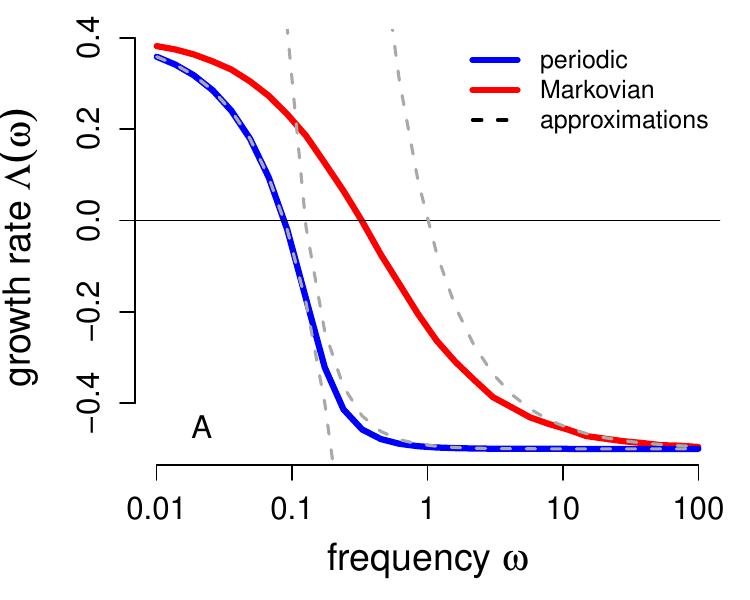}
\includegraphics[height=3in]{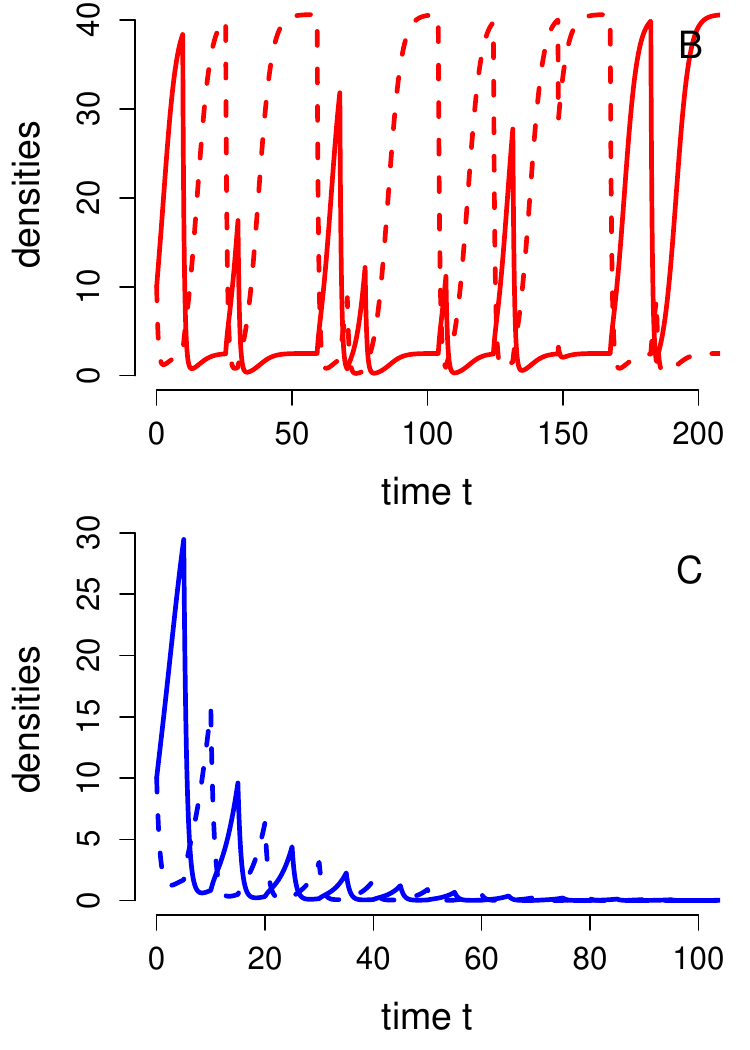}
\caption{ The Lyapunov exponents and density-dependent dynamics for a two patch metapopulation model~\eqref{eq:DIG-matrix}. In A, the Lyapunov exponent as a function of the switching frequency $\omega$. Solid lines correspond to numerically approximated Lyapunov exponents for periodic (blue) and random (red) switching. The dark gray, dashed lines correspond to the analytical approximations. Random (in C) and periodic (in C) simulations at frequency $\omega=0.1$ for a nonlinear model whose linearization at the origin is given by ~\eqref{eq:DIG-matrix}. Parameters: $L= \begin{pmatrix} -0.1 & 0.1 \\ 0.1 & -0.1 \end{pmatrix}$, $R_1=\begin{pmatrix} 0.5 & 0 \\ 0 & -1.5 \end{pmatrix}$, and $R_2=\begin{pmatrix} -1.5 & 0 \\ 0 & 0.5 \end{pmatrix}$. For B and C, the dynamics are $x'(t)=A_{\sigma(\omega t)} x(t)- 0.01 x(t)\circ x(t)$ where $\circ$ denotes the Hadamard product. }\label{fig:2patch}
\end{figure}

For a symmetric dispersal matrix and continuous (rather than piece-wise continuous), periodic $r^i:\R\to\R$ , \citet{K22}  proved the following result about monotonicity of the $\Lambda_p(\omega)$. 
\begin{thm}\citep{K22}\label{thm:K22}
Assume $L$ is irreducible and symmetric, $r^i:\R\to\R$ are continuous, period $1$ functions that are not all equal. Then the dominant Floquet multiplier $\Lambda_p(\omega)$ of 
\[
x_i'(t) = r^i(\omega t) + \sum_{j \neq i} l_{i,j} ( x_j - x_i).
\]  is a strictly decreasing function of $\omega$. 
\end{thm}
Theorem~\ref{thm:K22} suggests that slower switching promotes metapopulation persistence i.e. a higher growth rate $\Lambda_p(\omega)$. In particular, if $\sum_{i=1}^N \alpha_i \lambda(A_i)>0>\lambda(\sum_{i=1}^N \alpha_i A_i)$, then there is a critical frequency $\omega^*$ such that the metapopulation persists if and only if the frequency $\omega$ of environmental switching lies below $\omega^*$ (see black curve in Figure~\ref{fig:2patch})

Using Theorem~\ref{thm:general}, we show that Theorem~\ref{thm:K22} may extend to random and periodic, piecewise continuous switching and may partially extend to asymmetric dispersal matrices $L$. 

\begin{thm}\label{thm:inflation}
Assume $L$ is irreducible and the $R_i$ in \eqref{eq:DIG-matrix} are all not equal. 
\begin{enumerate}
\item (Fast switching)  $c_{fM}>c_{fp}= 0$.
\item (Slow switching) If $L$ is symmetric or $d=2$, then $c_{sp}\le 0$ and $c_{sM}\le 0$. Moreover, equality holds if and only if the right eigenvectors $x_i$ of $A_i=R_i +L$ are all equal. 
\end{enumerate}
\end{thm}

In the case of random switching, Theorem~\ref{thm:inflation} implies that sufficiently fast switching always leads to lower population growth rates. Specifically there exists $\omega^* \ge 0$ such that $\omega\mapsto \Lambda_{M}(\omega)$ decreases with the frequency $\omega$ of environmental switching for $\omega\ge \omega^*$.  This occurs for all forms of dispersal - symmetric or asymmetric. If dispersal is symmetric or there are only two patches, Theorem~\ref{thm:inflation} ensures that $\Lambda(\omega)$ decreases for sufficiently slow frequencies. Hence, we make the following conjecture for the random case:
\begin{conj}\label{conjecture}
Assume $L$ is irreducible and the eigenvectors $x_i$ of $A_i=R_i+L$ in \eqref{eq:DIG-matrix} are all not equal. If $L$ is symmetric or $d=2$, then    $\omega \mapsto \Lambda_M(\omega)$ is strictly decreasing. 
\end{conj}

Figure~\ref{fig:2patch} illustrates Theorems~\ref{thm:K22} and \ref{thm:inflation} and Conjecture~\ref{conjecture} numerically for a two-patch, two environment model ($d=N=2$). Both for random and periodic switching, the population growth rate $\Lambda(\omega)$ decreases with the switching frequency.  Hence, the critical frequency $\omega^*$ below which the metapopulation persists is higher for the randomly switching environment than the periodically switching environment.

What happens  when there are $d\ge 3$ patches and the dispersal matrix $L$ is asymmetric?  It turns out that  $c_{sM}$ and $c_{sp}$ can be positive and, consequently, $\omega\to\Lambda(\omega)$ need not be monotonic. To illustrate this possibility, consider patches lying along a circle and individuals moving clockwise along this circle. Specifically, let $L$ be the circular permutation matrix:
\[L = \begin{pmatrix}
-1 & 0 & \dots & 0 & 1 \\
1 & -1 & \dots & 0 & 0 \\
0 & 1 & \ddots & 0 & 0 \\
\vdots & \ddots & \ddots & \vdots & \vdots \\
0 & 0  & \dots & 1 & -1 
\end{pmatrix}\]
Namely $L_{i,i}=1$, $L_{1,d}=1$, $L_{i+1,i}=1$ for $1\le i\le d-1$, and $L_{i,j}=0$ otherwise. We assume there are $N=d$ environmental states and the per-capita growth of patch $i$ in environmental state $i$ equals $\rho+\beta\in \R$ and otherwise is $\rho\in \R$. Namely,  $r_i^i=\beta+\rho$ and $r_i^j=\rho$ for $i\neq j$. Finally, we consider the Markovian signal $\sigma$ with transition matrix $\Xi$ being either $L$ or $L^\intercal$. When $\Xi=L^\intercal$,  the signal moves in the same direction as the population (we call this the synchronized case), when $\Xi=L$ it goes backward (we call this the asynchronized case), see Figure~\ref{fig:migration}. As discussed at the end of Section~\ref{subsec:general},  the first order term expansion of $\Lambda_p(\omega)$ when $\sigma$ deterministically goes from $i$ to $i+1$ (synchronized case) or from $i$ to $i-1$ (asynchronized) is given by $c_{sp} = dc_{sM}$.

\begin{figure}
    \centering
    \includegraphics[scale=0.5]{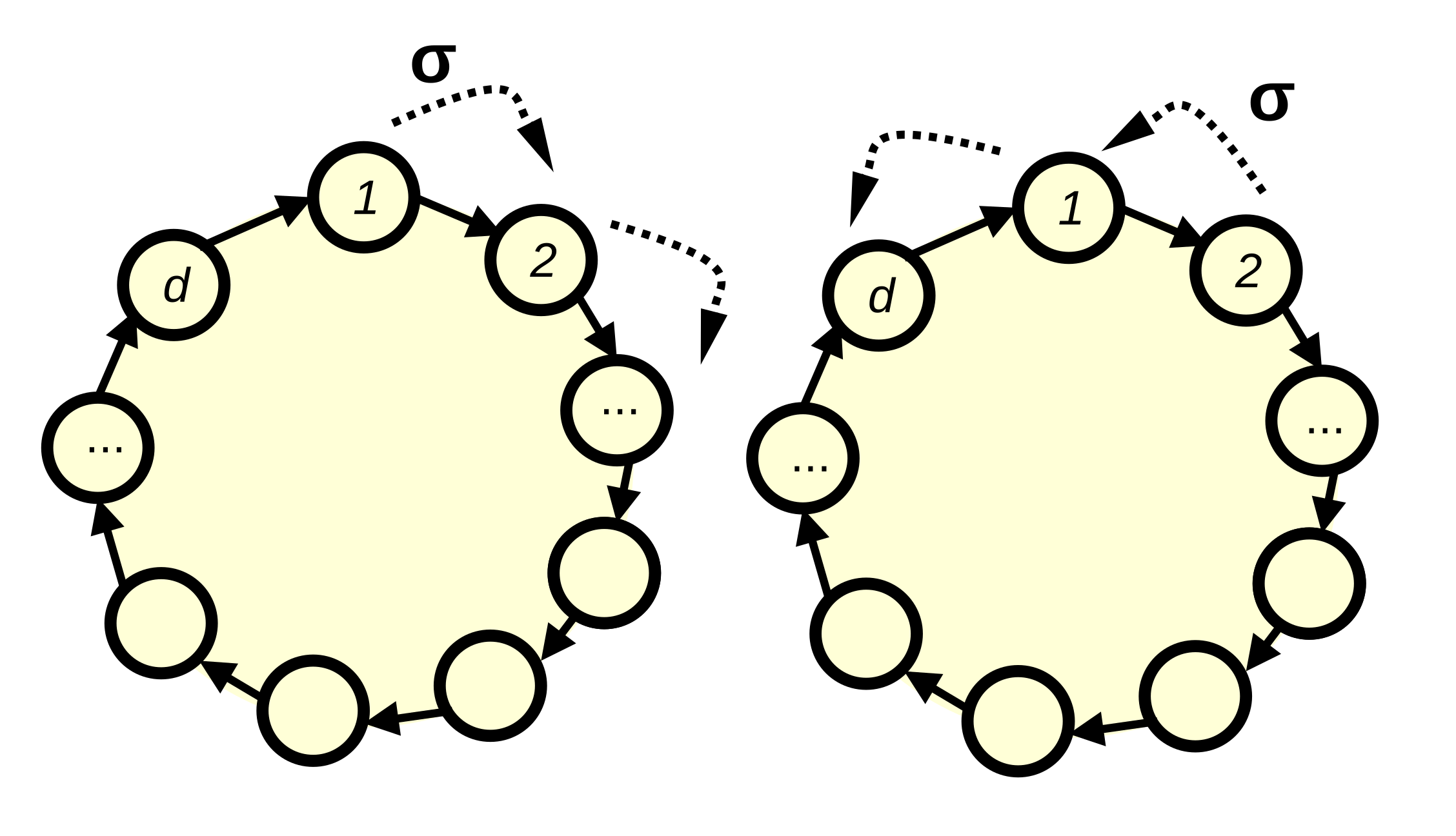}
    \caption{The plain arrows represent the population migration, the dotted arrows the random transitions of the signal, i.e. of the state where there is a non-zero growth rate.
    Left: $\Xi=L^\intercal$, the signal moves in the same direction as the population. Right: $\Xi=L$, the signal moves in the reverse direction.}
    \label{fig:migration}
\end{figure}

\begin{prop}\label{prop:circular}
Let $\eta$ be the unique positive solution of $(\eta-\beta) \eta^{d-1} = 1$. 
\begin{enumerate}
\item (synchronized case) If $\Xi=L^\intercal$, then 
\[c_{sM} = c_{sM}^{sync} :=     \ln \po\frac{(d-2)(\eta-\beta) +2 \eta }{(d-1)(\eta-\beta)\eta + \eta^2 }\pf \,. \]
In the case $d=3$, the sign of $c_{sM}$ is the opposite of the sign of $\beta$. For $d>3$, its sign is the sign of $\beta_d - \beta$ where $\beta_d = \eta_*-\eta_*^{1-d}>0$ with $\eta_*>2(d-1)/(d+1) > 1$ the unique root of $X^{d+1}-2X^d+(d-1)X-d+2$ over $(1,\infty)$.
\item (asynchronized case) If $\Xi=L$, then 
\[c_{sM} = c_{sM}^{async} :=   \ln \po\frac{ d \eta  }{d-1  + \eta^d }\pf\,,  \]
which is negative for all $\beta\neq 0$, $d\geqslant 3$.
\end{enumerate}
For the periodic case where $\sigma$ deterministically goes from $i$ to $i+1$ (synchronized case) or from $i$ to $i-1$ (asynchronized case),  the first order term expansion of $\Lambda_p(\omega)$ is given by $c_{sp} = dc_{sM}$.
\end{prop}

The Proof of Proposition~\ref{prop:circular} is in Section~\ref{sec:circular:proof}. Figure \ref{fig:Lambda_M(T)} illustrates the main conclusions of this proposition with $d=3$ patches, $\beta=-1$, and $\rho=0.3$. 
\begin{figure}
    \centering
    \includegraphics[width=0.45\textwidth]{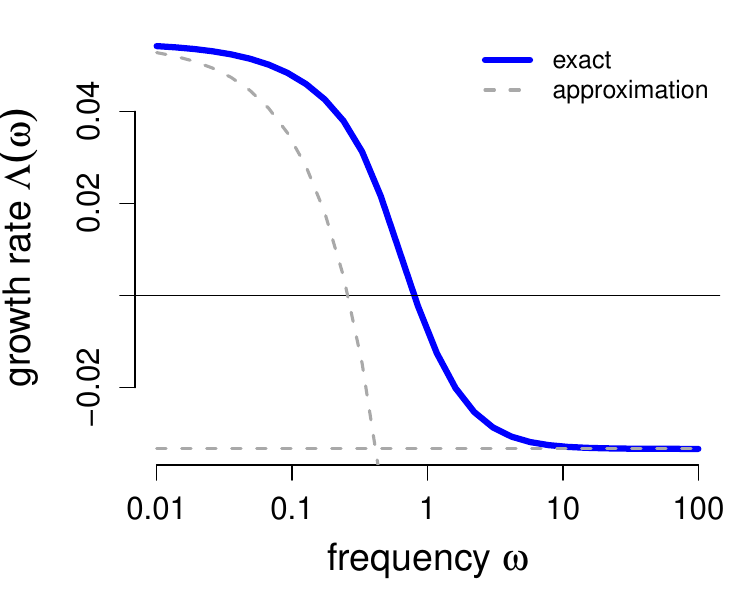} \includegraphics[width=0.45\textwidth]{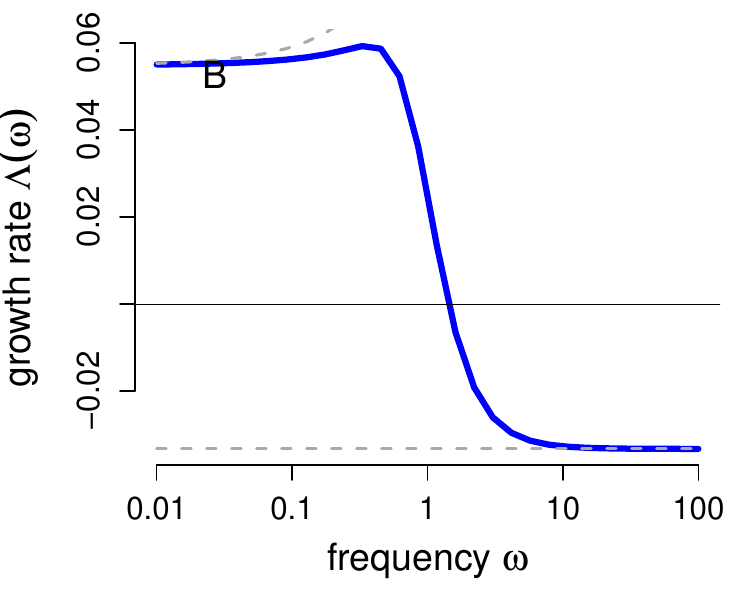}
\caption{The Lyapunov exponent $\Lambda_p(\omega)$ from the circular metapopulation model with $d=3$ patches. In panel A, the asynchronous case where $\Xi=L$. In panel B, the synchronous case where $\Xi=L^\intercal$. The per-capita growth $r_i^j$ in patch $i$ equals $-0.7$ in environment $j=i$ and $0.3$ otherwise.}
    \label{fig:Lambda_M(T)}
\end{figure}

\section{Discussion}\label{sec:discussion}

Populations experience environmental fluctuations at multiple time scales from diurnal cycles to multi-decadal climate cycles~\citep{hasegawa2022decadal,gorenstein202350}. Here, we derived analytical approximations for how the frequency of these environmental fluctuations  influence the long-term growth rate of structured populations. In the limit of low-frequency fluctuations, we derived new analytic approximations of the form $\Lambda(\omega)=\sum_i \alpha_i \lambda(A_i)+c\omega+o(\omega)$ with explicit formulas for $c$ in the periodic and random cases. In this low frequency limit, we showed that the mode of fluctuations, random versus periodic, has no effect on the population growth rate i.e. the correction terms $c_p$ and $c_M$ are equivalent when the random and periodic signal are comparable. This differs sharply from the high-frequency limit. In the limit of high-frequency fluctuations, we derived an  analytical approximation for periodic environments complementing the work of \citet{MonmarcheStrickler} for random environments.  For both approximations, the population growth rate in the infinitely fast switching limit are equal (i.e. $\lim_{\omega\to\infty}\Lambda_p(\omega)=\lim_{\omega\to\infty}\Lambda_M(\omega)=\lambda(\sum_i\alpha_i A_i)$), but the first-order correction terms with respect to the period $1/\omega$ of fluctuations, in general, differ (i.e. $c_{fp}\neq c_{fM}$). Indeed, for switching between two environments,  the first-order correction term is zero for the periodic case, but (in general) non-zero for the  random case. This differences implies that whether populations persist or not may  depend on the mode of environmental switching. 

\subsection*{Biological interpretations of the slow and fast limits}

The slow limit approximation has a clear biological interpretation in terms of the stable state distributions $x_i$ and the vectors of reproductive values $y_i$ -- two fundamental quantities of population modeling~\citep{caswell-01}. When environmental shifts are rare, the distribution of the population approaches the stable state distribution of the current environment, say $x_i$ in environment $i$. Reproductive values for this environment, the components of $y_i$, correspond to the  contributions of an individual in each state to the long-term population~\citep{caswell-01}. This interpretation of the absolute, rather than relative, contributions to the growth rate relies on the vector of reproductive values being normalized, so that, the expected value of the reproductive value of a randomly chosen individual from the stable state distribution is one i.e. $y_i^T x_i=1$~\citep{caswell-01}. When the environment shifts to a new environment, say $j$, the expected reproductive value of a randomly chosen individual shifts to $y_j^T x_i$ i.e. the expected reproductive value in the new environment of a randomly chosen individual from the old environment. When this new expectation is greater than one, individuals, on average, have greater reproductive value immediately after the environmental shift. This increase in reproductive value boosts the population growth rate by $\log(y_j^T x_i)$. If environmental shifts, on average, boost reproductive values, then the long-term population growth rate $\Lambda(\omega)$ increases with frequency $\omega$-- at least for $\omega$ sufficiently low. This phenomena is illustrated with our model~\eqref{eq:JA1} of a population of juveniles and adults where birth rates fluctuate between low and high values. In low birth rate environments, the stable state distribution mostly consists of adults and adults have slightly higher reproductive values. In the high birth rate environments there are mostly juveniles in the stable state distribution but adults have much high reproductive value. Hence, rare shifts from the low birth rate to high birth rate environments substantially increase the average reproductive value of an individual. In contrast, rare shifts from the high birth rate to low birth rate environments slightly decrease the average reproductive value of an individual. Hence, the net effect of environmental shifts is positive and the long-term population growth rate $\Lambda(\omega)$ increases with frequency. In contrast, when the juvenile and adult  death rates are fluctuating  i.e. equation~\eqref{eq:JA2}, we get the opposite trend. In the environment with higher adult mortality, there are more juveniles at the stable state distribution and juveniles have higher reproductive value. In the environment with higher juvenile adult mortality, adults are more common at the stable state distribution and adults have higher reproductive value. Hence, switching between environments always results in a decrease in the average reproductive value of an individual and a corresponding decrease in the population growth rate. Similar  reasoning applies to the two patch metapopulation model in Figure~\ref{fig:2patch}: switching always results in more individuals in the lower quality patch and a reduction in the population growth rate. 

The  interpretation of the fast limit approximations are less clear in general. However, we can draw some  conclusions in the case of random switching between two environments. When only the contribution of one population state to another population state fluctuates (i.e. a non-diagonal entry of $A$), Proposition~\ref{prop:N2Markov} implies that higher frequency fluctuations always increase the long-term population growth rate. In contrast, when only the contribution of population states to themselves rapidly fluctuate (i.e. only diagonal terms fluctuate), higher frequencies decrease the long-term population growth rate. In particular, this conclusion applies to metapopulations with constant dispersal rates whose per-capita growth rates fluctuate between two values.

\subsection*{Implications for metapopulations}
For metapopulations with symmetric dispersal matrices, \citet{K22} showed that the long-term population growth rate $\Lambda_p(\omega)$
always decreases with frequency when per-capita growth rates exhibit periodic, continuous fluctuations. Hence, metapopulations exhibiting diffusive-like movement grow faster in low frequency, periodic environments. Here, we showed the same conclusion holds for randomly fluctuating environments in the limits of fast and slow switching (see Theorem~\ref{thm:inflation}). Moreover, in the fast switching limit, long-term population growth rates are higher in randomly switching environments than periodically switching environments. An open problem (see Conjecture~\ref{conjecture}) is whether these conclusions hold at intermediate frequencies of random switching. When metapopulations exhibit asymmetric dispersal, however, we show that the long-term population growth rate may not vary monotonically with the environmental frequency $\omega$. For example, we constructed an example (see Figure~\ref{fig:Lambda_M(T)}) in which $\Lambda(\omega)$ is maximized at intermediate environmental frequencies. A likely hard open problem is characterizing which asymmetric dispersal matrices result in similar behaviors.

\begin{figure}
\center \includegraphics[width=0.45\textwidth]{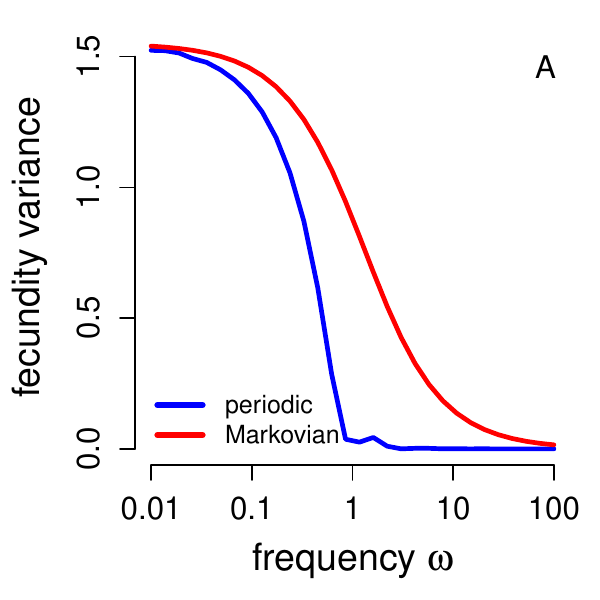}\includegraphics[width=0.45\textwidth]{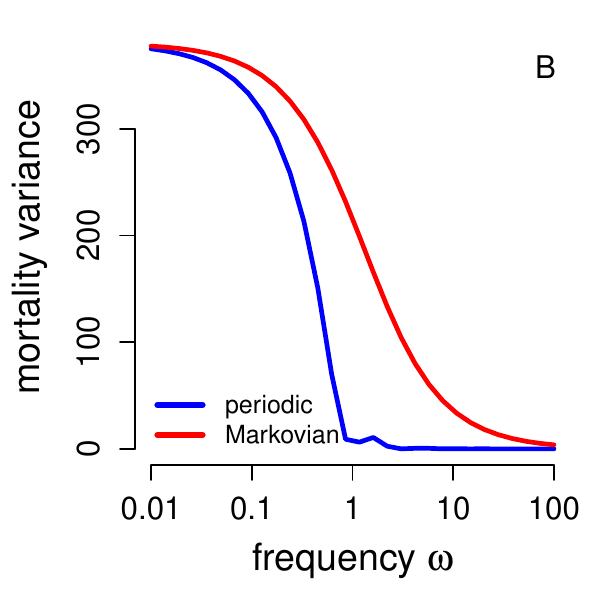}
\caption{Demographic variances \eqref{eq:varM} and \eqref{eq:varp} for the stage structured models \eqref{eq:JA1}-\eqref{eq:JA2} with fluctuating birth rates (A) and death rates (B). Parameters are as in Figure~\ref{fig:AJ}}\label{fig:variances}
\end{figure}

\subsection*{Relationship to stochastic demography}

The effect of temporal variation in environmental conditions on the growth rate of structured population is a primary focus of the field of stochastic demography~\citep{tuljapurkar-90,caswell-01,boyce-etal-06,hastings2012encyclopedia}. This extensive body of work, however, focuses exclusively on discrete-time models. Hence, to make comparisons to this literature requires defining an appropriate metric of variation of instantaneous demographic rates, say an entry $a(\omega t)$ of the matrix $A_{\sigma(\omega t)}$. One natural option is to consider the variation in the average demographic rate at stationarity over a fixed time interval, say of length one. In a random environment with stationary distribution $\pi$, this metric of variation equals
\begin{equation}\label{eq:varM}
\begin{aligned}
\mbox{Var}_M(\omega)=&\mathbb{E}_\pi \left[ \left(\int_0^1 a(\omega t)dt\right)^2 \right]- \mathbb{E}_\pi \left[\int_0^1 a(\omega t)dt\right]^2\\
=& \mathbb{E}_\pi \left[ \left(\frac{1}{\omega}\int_0^\omega a( t)dt\right)^2 \right]- \mathbb{E}_\pi \left[\frac{1}{\omega}\int_0^\omega a(t)dt\right]^2
\end{aligned}
\end{equation}
where $\mathbb{E}_\pi$ corresponds to the expectation with $\mathbb{P}[\sigma(0)=i]=\pi_i$. $\mbox{Var}_M(\omega)$ is a smooth function of $\omega$ and converges to zero as $\omega\to 0$ (see, e.g. \citep{MonmarcheStrickler}). Similarly, one can define variation of $a$ in the periodic environment as
\begin{equation}\label{eq:varp}
\mbox{Var}_p(\omega)= \int_0^1\left[ \left(\frac{1}{\omega}\int_0^\omega a(s+t)dt\right)^2 \right]ds-\left[\int_0^1\frac{1}{\omega}\int_0^\omega a(t+s)dtds\right]^2
\end{equation}
where $a(t)$ is started randomly in the interval $[0,1]$. $\mbox{Var}_p(\omega)$ also converges to zero as $\omega\to\infty$. However, unlike the random case, this convergence isn't monotone as $\mbox{Var}_p(\omega)$ equals zero at integer values. None the less, as illustrated in Figure~\ref{fig:variances}, average birth and death rates for the juvenile-adult models tend to decrease with frequency for the periodic case and continually decrease for random fluctuations. These trends facilitates comparisons to the stochastic demography literature. 

A central dogma of stochastic demography is \citep[page 96]{hastings2012encyclopedia} ``temporal variation in the vital rates will typically lead to a reduction in population growth rate.''  This genesis of this dogma were discrete-time geometric random walks and the observation that geometric mean is less than the arithmetic mean~\citep{lewontin-cohen-69}. For random, discrete-time matrix models $x(t+1)=A(t+1)x(t)$, \citet{tuljapurkar-90} extended this dogma provided the non-negative matrices $A(1),A(2),\dots$ are serially uncorrelated: the dominant Lyapunov exponent of this random product decreases with variation in the matrix entries. The continuous-time models considered here, however, are inherently temporally autocorrelated over sufficiently short time scales. Consistent with our results, the central dogma of stochastic demography doesn't apply, in general, to discrete-time models with temporal auto-correlations: temporally auto-correlated fluctuations can increase population growth~\citep{tuljapurkar-90}. This ``inflationary effect'' of auto-correlated fluctuations have been observed in stage-structured~\citep{tuljapurkar2006temporal} and spatially structured models~\citep{roy-etal-05,tuljapurkar2006temporal,prsb-10,kortessis-etal-2023,schreiber2023partitioning}. Consistent with these earlier studies, we found fluctuations in mortality for stage-structured models and per-capita growth rates in spatially-structured models can generate this inflationary effect: lower frequencies of fluctuations generate greater variation in demographic rates averaged over a time step and higher population growth rates. Notably, our approximations suggest that this inflationary effect in spatially-structured population models  occurs whenever there is symmetric dispersal.

\subsection*{Concluding remarks}

Our approach of modeling environmental fluctuations as a piece-wise constant, continuous process offers an analytically more tractable method for assessing the impact of environmental variability on the long-term growth rates of structured populations. The analytical approximations we've developed for growth rates in the limits of slow and fast environmental switching  may serve as valuable new tools for theoretical biologists and applied mathematicians. These tools enable a deeper examination of how both the tempo (slow versus fast) and mode (periodic versus random) of environmental fluctuations influence critical population dynamics, including growth, persistence, and extinction risk. This refined approach not only enhances our understanding of population responses to environmental variability but also opens new avenues for research in population biology and applied mathematics with potential applications to epidemiology, ecology, population genetics, and  conservation biology.

\section{Proof of  Theorem~\ref{thm:general}}\label{sec:proofs_general}

\subsection{Slow periodic switching}

We prove formula~\eqref{eq:thm_general_slow_periodic}. Recall that $T = \omega^{-1}$. We write $r_i= \mu(A_i)$. By Perron-Frobenius Theorem, for all $i$, for all $T > 0$,
\[
e^{-r_i \alpha_i T}e^{ \alpha_i T A_i} = p_i + o(e^{- \gamma T}),
\]
where $\gamma > 0$ and $p_i = x_i y_i^\intercal$ is the projection matrix on $x_i$ in the direction of $y_i$. Thus,
\[
 \prod_{i=1}^N e^{-r_i \alpha_i T} e^{T \alpha_i A_i} = \prod_{i=1}^N p_i + o(e^{-\gamma T}), 
\]
which leads to
\[
\lambda(M(T)) = \prod_{i=1}^N e^{r_i \alpha_i T} \lambda \left( \prod_{i=1}^N p_i + o(e^{-\gamma T}). \right)
\]
Now, $ \lambda \left( \prod_{i=1}^N p_i + o(e^{-\gamma T})\right) = \lambda \left( \prod_{i=1}^N p_i \right) + o(\frac{1}{T})$ and 
\[
\prod_{i=1}^N p_i = x_N y_N^\intercal \ldots x_1 y_1^\intercal =  \left( \prod_{i=1}^{N-1} x_i^\intercal y_{i+1} \right)  x_N y_1^\intercal.
\]
In particular, $\lambda( \prod_{i=1}^N p_i ) = \prod_{i=1}^{N-1}x_i^\intercal y_{i+1} \lambda( x_N y_1^\intercal)$, where $x_N y_1^\intercal$ is the projection matrix on $x_N$ in the direction of $y_1^\intercal$. Hence, $\lambda( x_N y_1^\intercal) =  x_N^\intercal y_1$. This concludes the proof by Lemma \ref{lem:lambdaperiod}.

\subsection{Fast periodic switching}

This section is devoted to the proof of Formula~\eqref{eq:thm_general_fast_periodic}. On the one hand,  for an irreducible Metzler matrix $B$, we have 
\[
\lambda(e^B) = e^{\lambda(B)}.
\]
Indeed, since $B$ is Metzler and irreducible, $e^B$ is a  nonnegative irreducible matrix, hence $\lambda(e^B) = r(e^B)$. Now it is easily seen that $r(e^B) = e^{\lambda(B)}$ because the spectrum of  $e^B$ is the exponential of the spectrum of $B$. On the other hand, by an iteration of the Baker - Campbell - Hausdorff formula, one has
\[
M(T) = \exp \left( T \bar{A} + \frac{T^2}{2} \sum_{1 \leq i < j \leq N} [\alpha_j A_j, \alpha_i A_i] + o(T^2) \right)
\]
Thus, 
\begin{equation}
\frac{1}{T}\log \lambda(M(T)) = \lambda( B(T)),
\label{eq:mu}
\end{equation}
where $B(T) = \bar{A} + \frac{T}{2} \sum_{1 \leq i < j \leq N} \alpha_j \alpha_i [A_j, A_i] + o(T)$. Now, since $\bar A$ has a unique dominant eigenvalue, the same holds for $B(T)$, for $T$ small enough. Moreover, one has $B'(0) = \frac{1}{2}\sum_{1 \leq i < j \leq N} \alpha_j \alpha_i [A_j, A_i]$. Thus, by \cite[Theorem 6.3.12]{HJ13} or \cite[Theorem 4.1]{HRR92}, the function $T \mapsto \lambda(B(T))$ is differentiable at $0$ and its derivative at $0$ is given by $\bar y^\intercal B'(0) \bar x = \bar y^\intercal \left( \frac{1}{2}\sum_{1 \leq i < j \leq N} \alpha_j \alpha_i [A_j, A_i]\right)\bar x$. This, together with Lemma \ref{lem:lambdaperiod} and Equation \eqref{eq:mu}, proves \eqref{eq:thm_general_fast_periodic}.

\subsection{Slow Markovian switching}\label{subsec:proof_slow_Markov}

In this section, devoted to the proof of \eqref{eq:thm_general_slow_Markov}, $\sigma$ is a Markov chain as in Section~\ref{subsec:Markov}. We use the notations of Section~\ref{subsec:Markov}. In view of \eqref{eq:LyapunovSlowMarkov},
 the goal is  now to give an expansion of $\mu_\omega$ when $\omega$ vanishes.

In the following, for $i\in\cco 1,N\ccf$, we denote by $(\varphi_t^i)_{t\geqslant 0}$ the flow on $\Delta$ associated to $F_i$; namely, for $z\in\Delta$, $\theta_i(t) := \varphi_t^i(z)$ solves $\theta_i'(t) = F_i(\theta_i(t))$ for all $t\geqslant 0$ with $\theta_i(0)=z$. Recall that $x_i$ stands for a unit right eigenvector of $A_i$. The following follows e.g. from \cite[Lemma 6]{benaim2019}.

\begin{lem}\label{lem:slowMarkov}
Assuming that $A_1,\dots,A_N$ are Metzler and irreducible,  there exist $C,a>0$ such that for all $z\in\Delta$, $i\in\cco 1,N\ccf$ and $t\geqslant 0$,
\begin{equation}
    \label{eq:contractionphi_i}
|\varphi_t^i(x) - x_i| \leqslant C e^{-at}.
\end{equation}
\end{lem}

In the slow switching regime, since the signal $\sigma(\omega t)$ stays constant for very long times, Lemma~\ref{lem:slowMarkov} shows that the process $\theta$ in \eqref{eq:rhotheta} spends most of its time close to the points $(x_i)_{i\in\cco 1,N\ccf}$, and $\mu_\omega$ is then expected to be close to
\[\tilde \mu_0 = \sum_{i=1}^N \alpha _i \delta_{x_i,i}\,.\]
The next statement provides the first order expansion of this convergence. Notice that its proof doesn't rely on the specific form of the vector fields \eqref{eq:Fi}, but only on the contraction established in Lemma~\ref{lem:slowMarkov} (so that the next result can be straighforwardly extended to slow swithching between any contracting flows -- i.e. assuming \eqref{eq:contractionphi_i} --  over a compact manifold)  

\begin{prop}
\label{prop:expansion_mu_slow}
Assuming that $A_1,\dots,A_N$ are Metzler and irreducible, for all $f\in\mathcal C^1(\Delta\times\cco 1,N\ccf)$,
\[\mu_\omega f = \tilde \mu_0 f + \omega \tilde c_1 + \underset{\omega \rightarrow0} o(\omega)\]
with
\[ 
    \tilde c_1 = \sum_{i,j=1}^N \alpha_i q_{ij} \int_0^\infty \po f\po \varphi_s^j(x_i),j\pf - f(x_j,j)\pf     \dd s \,.
\]
\end{prop}

\begin{proof}
Fix $f\in\mathcal C^1(\Delta\times\cco 1,N\ccf)$.  Denote by $(T_k)_{k\in\N}$ the jump times of the Markov chain $(\sigma(\omega t))_{t\geqslant 0}$ and by $\nu_\omega $ the invariant measure of the discrete-time chain $(\theta({T_k}),\sigma(\omega {T_k}))_{k\in\N}$,  called the skeleton chain associated to $(\theta(t),\sigma(\omega t))_{t\geqslant 0}$. Recall the following relation between $\nu_\omega$ and $\mu_\omega$ (see \cite[Theorem 34.31 p.123]{DavisBook}): 
\begin{equation}
    \label{eq:mu_skeleton}
\mu_\omega  f =  \int_{\Delta \times\cco 1,N\ccf} \int_0^\infty f\po \varphi_s^i(z),i\pf \omega  e^{-\omega s}  \dd s \nu_\omega (\dd z\dd i) \,.
\end{equation}
Thanks to Lemma~\ref{lem:slowMarkov}, for any $(z,i) \in \Delta \times\cco 1,N\ccf$ we can write
\begin{align}
\int_0^\infty f\po \varphi_s^i(z),i\pf \omega  e^{-\omega  s}  \dd s & =  f(x_i,i) + \omega  \int_0^\infty \po f\po \varphi_s^i(z),i\pf - f(x_i,i)\pf    e^{-\omega  s}  \dd s \nonumber\\
 & =  f(x_i,i) + \omega  \int_0^\infty \po f\po \varphi_s^i(z),i\pf - f(x_i,i)\pf     \dd s + \underset{\omega\rightarrow 0} o(\omega)\,,\label{eq:demo_slow_Markov}
\end{align}
by dominated convergence. On the other hand, since $\nu_\omega$ is invariant for the Markov kernel $\mathcal T$ given by
\[ \mathcal T f(z,i) = \sum_{j=1}^N q_{ij} \int_0^\infty f\po \varphi_s^i(z),j\pf  \omega  e^{-\omega  s}\dd s,\]
using that $\nu_\omega = \nu_\omega \mathcal T$, that the marginal on $\cco 1,N\ccf$ of $\nu_\omega$ is $\alpha $ and that, thanks to Lemma~\ref{lem:slowMarkov},
\[\left|\mathcal T f(x,i) - \sum_{j=1}^N q_{ij}   f\po x_i,j\pf  \right| \leqslant C \|\na_x f\|_\infty \omega  \] 
for some $C>0$ independent from $\omega$, we immediately get
\[\nu_\omega  f \underset{\omega \rightarrow0}\longrightarrow \sum_{i,j=1}^N \alpha_i q_{ij} f(x_i,j) \,.\]
Using this after \eqref{eq:demo_slow_Markov} (and, again, that the marginal of $\nu_\omega$ on $\cco 1,N\ccf$ is $\alpha$) in 
 \eqref{eq:mu_skeleton}  we end up with
 \begin{eqnarray*}
 \mu_\omega f & = &   \tilde \mu_0 f + \omega \int_{\Delta \times\cco 1,N\ccf} \int_0^\infty \po f\po \varphi_s^j(z),i\pf - f(x_j,j)\pf     \dd s \nu_\omega (\dd z\dd i)  + \underset{\omega\rightarrow 0} o(\omega) \\
 &= & \tilde \mu_0 f + \omega  \sum_{i,j=1}^N \alpha_i q_{ij} \int_0^\infty \po f\po \varphi_s^j(x_i),j\pf - f(x_j,j)\pf     \dd s + \underset{\omega\rightarrow 0}o(\omega).
 \end{eqnarray*}
 This concludes the proof.
\end{proof}

In view of \eqref{eq:LyapunovSlowMarkov} and Proposition~\ref{prop:expansion_mu_slow}, the proof of the expansion \eqref{eq:thm_general_slow_Markov} in the slow Markovian switching case is concluded by the following.

\begin{lem} In Proposition~\ref{prop:expansion_mu_slow}, for $f(z,i) = \1 \cdot A_i z $, we have 
\[
\tilde c_1 = \sum_{i,j=1}^N \alpha _i q_{ij} \ln (y_j^\intercal x_i).
\]
\end{lem}
\begin{proof}
Notice that, for all $z \in \Delta_d$,
\[
  f\po \varphi_s^j(z),j\pf  = \frac{\1_d^\intercal A_j e^{s A_j}z}{\1_d^\intercal e^{s A_j}z}  
   = \frac{d}{ds} \ln( \1_d^\intercal e^{s A_j}z)\,.
\]
Theorefore, using that $x_j = \varphi_s^j(x_j)$, one has
\begin{align*}
    \int_0^\infty \po f\po \varphi_s^j(x_i),j\pf - f(x_j,j)\pf     \dd s & =   \int_0^\infty \frac{d}{ds} \po \ln \po \frac{\1_d^\intercal e^{s A_j}x_i}{\1_d^\intercal e^{s A_j}x_j} \pf \pf \dd s\\
    & = \lim_{s \to \infty} \ln \po \frac{\1_d^\intercal e^{s A_j}x_i}{\1_d^\intercal e^{s A_j}x_j} \pf\\
    & = \ln ( y_j^\intercal x_i ),
\end{align*}
where we have used (once again) Perron-Frobenius Theorem. 
\end{proof}

\section{Proof of Theorem~\ref{thm:inflation}}

We first prove a more general result about random switching of Metzler matrices that are symmetric. This result applies to the special case of \eqref{eq:DIG-matrix} when the dispersal matrix $L$ is symmetric. 

\begin{prop}
\label{prop:slow-symmetric}
In the settings of Theorem~\ref{thm:general}, assume that all the matrices $A_1,\dots,A_N$ are symmetric. Then, in \eqref{eq:thm_general_slow_periodic} and \eqref{eq:thm_general_slow_Markov}, $c_{sp} \leqslant 0$ and $c_{sM}\leqslant 0$. Moreover, either all the $x_i$ are equal, in which case $\omega \mapsto \Lambda_p(\omega)$ and $\omega \mapsto \Lambda_M(\omega)$ are constant, or otherwise $c_{sp} < 0$ and $c_{sM} < 0$.
\end{prop}

\begin{proof}
We only give the proof for the Markov case, since in the slow switching regime, $c_{sp}$ is positively proportional to   $c_{sM}$ for a particular Markov chain, as described at the end of Section~\ref{subsec:general}. Since $A_i$ is symmetric, its left eigenvector $y_i$ is proportional the right eigenvector $x_i$. Moreover, due to the normalization condition $y_i^\intercal x_i = 1$, we have $y_i = \|x_i\|^{-2} x_i$. Therefore, 
\begin{align*}
    \prod_{i,j} \left( y_j^\intercal x_i \right)^{q_{i,j} \alpha_i}  & = \prod_{i,j} \left( \|x_j\|^{-2} x_j^\intercal x_i \right)^{q_{i,j} \alpha_i}\\
    & \leq \prod_{i,j} \left( \frac{\|x_i\|}{\|x_j\|}\right)^{q_{i,j} \alpha_i},
\end{align*}
where the last line follows from Cauchy-Schwarz inequality. Now, in this last product, the exponent of each $\|x_i\|$ is $\sum_{j} \alpha_i q_{i,j} - \sum_{j} \alpha_j q_{j,i}$ , which is zero since $\alpha$ is invariant for $Q$. This entails that $c_{sM} \leq 0$. Moreover, $c_{sM} = 0$ if and only if there is equality in all the Cauchy-Schwarz inequalities we used, namely, if and only if all the $x_i$ are positively correlated, and thus, equal. In that case, one easily checks that $\mu_\omega$, the  invariant probability measure of the PDMP on the sphere considered in Section~\ref{subsec:proof_slow_Markov} for the proof of \eqref{eq:thm_general_slow_Markov}, is $\delta_{x_1} \otimes \alpha$ for all values of $\omega$. In view of \eqref{eq:LyapunovSlowMarkov},  $\omega \mapsto \Lambda_M(\omega)$ is thus constant. In the periodic case, $\Lambda_p(\omega)$ can also be written as the integral of a function of the unique globally asymptotically periodic solution of the system on the sphere (see Formula 29 in Theorem 8 in \cite{BLSS24}). Here again, it is easy to check that the constant solution equal to $x_1$ is the unique periodic solution, and thus $\Lambda_p(\omega)$ is constant.
\end{proof}

The next two propositions provide show that $c_{fM}>c_{fp}=0$ in the fast-switching limit.  

\begin{prop}
In the case \eqref{eq:DIG-matrix},  $c_{fp} = 0$.
\end{prop}

\begin{proof}
Recall that 
\[
c_{fp} = \frac{1}{2} \sum_{i < j} \alpha_i \alpha_j \bar y^{\intercal} [A_j, A_i] \bar x.
\]
Since $A_i = R_i +L$ and $R_i$ and $R_j$ commute, we end up with $[A_j, A_i] = [L, R_i] + [R_j, L]$. Now, 
\begin{eqnarray*}
\bar y^{\intercal} [L, R_i] \bar x  & =& \bar y^{\intercal} L R_i \bar x - \bar y^{\intercal} R_i L \bar x \\
&= & \bar y^{\intercal}( \lambda(\bar A) - \bar R)R_i \bar x -  \bar y^{\intercal}R_i(  \lambda(\bar A) - \bar R) \bar x   \\
 & = &   0\,,
\end{eqnarray*} 
where we have used that $\bar R$ and $R_i$ commute and that by definition, 
\[
\bar y^{\intercal}( \bar R + L) = \lambda(\bar A) \bar y^{\intercal}, \quad ( \bar R + L) \bar x = \lambda(\bar A) \bar x .
\]
\end{proof}

\begin{prop}
In the case \eqref{eq:DIG-matrix},  $c_{fM} > 0$.
\end{prop}

\begin{proof}
First, let us prove that, in the case \eqref{eq:DIG-matrix},  
\begin{equation}\label{eq:cfMALR}
  c_{fM} =  \sum_{i,j} \alpha_i (Q-I)_{i,j}^{-1} \left(\bar y^{\intercal} R_j \bar x \bar y^{\intercal} R_i \bar x - \bar y^{\intercal} R_j R_i \bar x\right)\,.  
\end{equation}
Indeed, 
\[
    \bar y^{\intercal}A_j \bar x \bar y^{\intercal} A_i \bar x - \bar y^{\intercal} A_j A_i \bar x  = \bar y^{\intercal} R_j \bar x \bar y^{\intercal} R_i \bar x - \bar y^{\intercal} R_j R_i \bar x + B_i + C_j + D,
\]
where 
\begin{align*}
    B_i  &= \bar y^{\intercal} L \bar x \bar y^{\intercal} R_i \bar x - \bar y^{\intercal} L R_i \bar x \\
    C_j  &= \bar y^{\intercal} R_j \bar x \bar y^{\intercal} L \bar x - \bar y^{\intercal}  R_j L \bar x\\
    D &= \bar y^{\intercal} L \bar x \bar y^{\intercal} L \bar x - \bar y^{\intercal} L^2 \bar x\,.
\end{align*}
Using that $\sum_i \alpha_i (Q-I)_{ij}^{-1} = 0= \sum_j (Q-I)_{i,j}^{-1} $, we end up with $\sum_{i,j} \alpha_i (Q-I)_{i,j}^{-1} (B_i + C_j + L) = 0$, leading to \eqref{eq:cfMALR}.

For $k\in\cco 1,N\ccf$, write $a_k = \bar x_k \bar y_k$ and $(R_i)_{k,k} =r^{i,k}$. Since the $R_i$'s are diagonal, we get that $\bar y^{\intercal} R_i \bar x = \sum_k a_k r^{i,k}$. Now, let $X = (X_1, \ldots, X_N)$ be a random variable taking the value  $(r^{1,k}, \ldots r^{N,k})$ with probability $a_k$ for all $k\in\cco 1,N\ccf$ and define $\varphi : \mathbb{R}^N \to \mathbb{R}$ by 
\[
\varphi: (x_1, \ldots, x_N) \mapsto - \sum_{i,j}  \alpha_i (Q-I)_{i,j}^{-1} x_i x_j.
\]
It is readily checked that 
\[
\sum_{i,j} \alpha_i (Q-I)_{i,j}^{-1} \left(\bar y^{\intercal} R_j \bar x \bar y^{\intercal} R_i \bar x - \bar y^{\intercal} R_j R_i \bar x\right) = \mathbb E( \varphi(X) ) - \varphi(\mathbb E( X)).
\]
It remains to prove that $\varphi$ is convex, which will imply by Jensen inequality that $c_{fM} > 0$. Since $\varphi$ is a quadratic form, it is convex if and only if it is positive. Yet, $\varphi(x) = - \int f(i) (Q-I)^{-1}f(i) \alpha(d i)$, where $f(i) = x_i$, which can be interpreted as an asymptotic variance, and is thus positive (see the discussion after Proposition 4 in \cite{MonmarcheStrickler}).
\end{proof}

Finally, we consider the case of $d=2$ patches. The matrices $A_i$ in case \eqref{eq:DIG-matrix} can be written as
\[
A_i = \begin{pmatrix}
a_i & \alpha\\
\beta & b_i
\end{pmatrix}, \quad i=1,\ldots, N.
\]
for some $a_i, b_i \in \mathbb{R}$ and $\alpha, \beta > 0$ (recall that the matrices are Metzler and irreducible). The right and left eigenvectors can be computed explicitely as follows
\begin{equation}
\label{eq:vecteurpopredim2}
    x_i = \frac{1}{2 \alpha + g_i}[ 2 \alpha, g_i]^{\intercal}, \quad y_i= \frac{2 \alpha + g_i}{4 \alpha \beta + g_i^2}[ 2 \beta, g_i]^{\intercal},
\end{equation}
with
\[
g_i = \sqrt{(b_i - a_i)^2 + 4 \alpha \beta} + b_i - a_i > 0.
\]
Note that for $i \neq j$, $g_i = g_j$ if and only if $x_i = x_j$, if and only if $A_i = A_j + \delta_i I$ for some $\delta_i \in \mathbb{R}$.
\begin{prop}
In case \eqref{eq:DIG-matrix}, when $d= 2$, 
\[
c_{sM} = \sum_{i,j} \alpha_i q_{ij} \ln \po  \frac{4 \alpha \beta + g_i g_j}{4 \alpha \beta + g_j^2} \pf \leq 0.
\]
Moreover, either for all $j > 1$, there is $\delta_j$ such that $A_j = A_1 + \delta_j I$ and $\omega \mapsto \Lambda_M(\omega)$ is constant and in particular, $c_{sM} = 0$; or otherwise $c_{sM} < 0$.
\end{prop}

\begin{proof}
By Equations~\eqref{eq:thm_general_slow_Markov} and~\eqref{eq:vecteurpopredim2}, we get that
\[
c_{sM} = \sum_{i,j} \alpha_i q_{ij} \ln \po  \frac{4 \alpha \beta + g_i g_j}{4 \alpha \beta + g_j^2} \frac{2 \alpha + g_j}{2 \alpha + g_i} \pf.
\]
Note that by invariance of $\alpha$ with respect to $Q$, we have for all $(\eta_i)_{1 \leq i \leq N}$
\begin{equation}
    \label{eq:doublesuminv}
    \sum_{i,j} \alpha_i q_{i,j} \eta_i = \sum_{i,j} \alpha_i q_{i,j} \eta_j = \sum_i \alpha_i \eta_i
\end{equation}
This yields
\[
c_{sM} = \sum_{i,j} \alpha_i q_{ij} \ln \po  \frac{4 \alpha \beta + g_i g_j}{4 \alpha \beta + g_j^2} \pf.
\]
Now, applying  for all $i$ and $j$ Cauchy-Schwarz inequality to the two-dimensional vectors $(\sqrt{4 \alpha \beta}, g_i)$ and $(\sqrt{4 \alpha \beta}, g_j)$, we have
\[
\po 4 \alpha \beta + g_i g_j \pf \leq \po 4 \alpha \beta + g_i^2 \pf^{1/2} \po 4 \alpha \beta + g_j^2 \pf^{1/2},
\]
with equality if and only if $g_i = g_j$. Therefore, 
\begin{align*}
    \sum_{i,j} \alpha_i q_{ij} \ln \po  4 \alpha \beta + g_i g_j \pf & \leq  \sum_{i,j} \alpha_i q_{ij} \ln \po \po 4 \alpha \beta + g_i^2 \pf^{1/2} \po 4 \alpha \beta + g_j^2 \pf^{1/2} \pf\\
    & = \frac{1}{2}\sum_{i,j} \alpha_i q_{ij} \ln  \po 4 \alpha \beta + g_i^2 \pf + \frac{1}{2}\sum_{i,j} \alpha_i q_{ij} \ln  \po 4 \alpha \beta + g_j^2 \pf\\
    & = \sum_{i,j} \alpha_i q_{ij} \ln  \po 4 \alpha \beta + g_i^2 \pf,
\end{align*}
where the last equality comes from~\eqref{eq:doublesuminv}, and the first inequality is an equality if and only if $g_i = g_j$ for all $i$ and $j$. Hence, we get that $c_{sM} \leq 0$, and $c_{sM} < 0$ except if all the $g_i$ are equal, that is, if and only if all the $x_i$ are equal. But, as explained in the proof of Proposition~\ref{prop:slow-symmetric}, this implies that the map $\omega \mapsto \Lambda_M(\omega)$ is constant.
\end{proof}

\section{Proof of Proposition~\ref{prop:circular}}\label{sec:circular:proof}
Because of the symmetries of the problem, the principal right eigenvector $x_i$ of $A_i=L +R_i$ is simply $x_1$ up to a circular relabelling of the 
states, i.e. $x_{i,k} = x_{1,k-i+1}$ for all $k\in\cco 1,d\ccf $ (recall that indexes are understood modulo $d$). Similarly, denoting by $\hat y_i$ the principal left eigenvector of $A_i$ normalized so that $\1_d^\intercal \hat y_i = 1$ (so that $y_i =( x_i^\intercal\hat y_i)^{-1} \hat y_i $), which is the right eigenvector of $R_i+L^T$ (similar to $R_i+L$ except that migration turns in the other direction), it is obtained from $x_i$ by $\hat y_{i,k} = x_{i,2i-k}$. Then, in \eqref{eq:thm_general_slow_Markov}, in the synchronized case,  
    \[c_{sM}=c_{sM}^{sync} := \frac1d  \sum_{i=1}^d   \ln (y_{i+1}^\intercal x_i)  =   \ln (y_2^\intercal x_1) =  \ln (\hat y_2^\intercal x_1) - \ln (\hat y_1^\intercal x_1) \,, \]
    while, in the a-synchronized case,
    \[c_{sM} = c_{sM}^{async} := \frac1d  \sum_{i=1}^d   \ln (y_{i-1}^\intercal x_i)  =  \ln (\hat y_d^\intercal x_1) - \ln (\hat y_1^\intercal x_1) \,, \]
    
    Denoting for brevity $\lambda = \lambda(A_1)$, $x_1$ solves 
    \[(\beta -1) x_{1,1} + x_{1,d} = \lambda x_{1,1}\,,\qquad - x_{1,i}+ x_{1,i-1} = \lambda x_{1,i} \quad \forall i \in\cco 1,d\ccf\,,  \]
    which, writing $r=\lambda+1$, is solved as
    \[ x_{1,i} =(r-\beta) r^{d-i} x_{1,1} \quad \forall i\in\cco 2,d\ccf\]
    with $x_{1,1}$ being fixed by the normalization. Moreover, the equation $x_{1,1}= rx_{1,2}$ shows that $r$ is the unique positive solution of
    \begin{equation}\label{eq:r}
        (r-\beta) r^{d-1} = 1
    \end{equation}
    In particular, $r> \beta$ and  $\lambda=r-1$  is always larger than $-1$, and has the same sign as $\beta$ (since the solution for $\beta=0$ is $r=1$, and the solution of this equation is an increasing function of $\beta$).

Using the symmetries previously mentioned,
\begin{eqnarray*}
\hat y_2^\intercal x_1   = (r-\beta)^2 x_{1,1}^2 
\begin{pmatrix}
 r^{d-2} \\  r^{d-1} \\ 1 \\  \vdots \\  r^{d-4} \\  r^{d-3} 
\end{pmatrix}
\cdot 
\begin{pmatrix}
r^{d-1} \\ r^{d-2} \\ r^{d-3} \\  \vdots \\ r \\ 1 
\end{pmatrix}  
&= &  (r-\beta)^2 x_{1,1}^2 \co (d-2)r^{d-3} +2r^{2d-3}\cf \\&= &   x_{1,1}^2 \co (d-2)(r-\beta)r^{-2} +2 r^{-1}\cf \,,
\end{eqnarray*}

while 
\begin{eqnarray*}
\hat y_1^\intercal x_1 = (r-\beta)^2 x_{1,1}^2 
\begin{pmatrix}
r^{d-1} \\ 1 \\ r \\  \vdots  \\  r^{d-3} \\  r^{d-2}  
\end{pmatrix}
\cdot 
\begin{pmatrix}
r^{d-1} \\ r^{d-2} \\ r^{d-3} \\  \vdots \\ r \\ 1 
\end{pmatrix}   
&= &  (r-\beta)^2 x_{1,1}^2 \co (d-1)r^{d-2} +r^{2d-2}\cf \\
&= &   x_{1,1}^2 \co (d-1)(r-\beta)r^{-1} +1 \cf \,,
\end{eqnarray*}
and 
\begin{eqnarray*}
\hat y_d^\intercal x_1   = (r-\beta)^2 x_{1,1}^2 
\begin{pmatrix}
 1 \\  r \\  r^2 \\  \vdots \\  r^{d-2} \\  r^{d-1} 
\end{pmatrix}
\cdot 
\begin{pmatrix}
r^{d-1} \\ r^{d-2} \\ r^{d-3} \\  \vdots \\  r \\ 1 
\end{pmatrix}  
&= &  (r-\beta)^2 x_{1,1}^2d  r^{d-1}  \\
&= &   x_{1,1}^2  (r-\beta) d \,.
\end{eqnarray*}
We have thus obtained that
\begin{eqnarray*}
c_{sM}^{sync} = \ln \po \frac{\hat y_2^\intercal x_1}{\hat y_1^\intercal x_1} \pf  & = &  \ln \po\frac{(d-2)(r-\beta) +2 r }{(d-1)(r-\beta)r + r^2 }\pf \\
& = & \ln \po  \frac{ (d-2) r^{1-d} + 2 r    }{(d-1)r^{2-d}+ r^2}\pf  \\
& = & \ln \po  \frac{ (d-2)  + 2 r^d     }{(d-1)r+ r^{d+1}}\pf\,.
\end{eqnarray*}
This has the same sign as $f(r)$ where
\[f(s) = (d-2) + 2 s^d - (d-1)s - s^{d+1} \,.\]
When $d=3$, we see that
\[f(s) = -s^4 + 2 s^3 - 2 s + 1 = (1-s)^3 (1+s)\,,\]
which for $s>0$ has the same sign as $1-s$. As we saw, $1-r=-\lambda$ has the same sign as $-\beta$, and as a conclusion in the case $d=3$, the sign of $c_{sM}$ is the  opposite of the sign of $\beta$.

Consider the case $d>3$ (for $d=2$, $L$ is symmetric). Computing
\begin{align*}
    f'(s) &= 2d s^{d-1} -(d-1) -(d+1) s^d \\ 
    f''(s) &= 2d(d-1) s^{d-2}-d(d+1) s^{d-1}\,,
\end{align*}
we see that $f(1)=f'(1)=0$,  and $f''$ cancels exactly twice, at $0$ and $s_* = \frac{2(d-1)}{d+1} \in(1,2)$, and is positive  on $(0,s^*)$ and negative on $(s^*, + \infty)$. In particular, $f$ is convex decreasing positive over $(0,1)$, increasing convex over $(1,s_*)$ and then it becomes strictly concave, being eventually decreasing and then going to $-\infty$. In particular, it admits exactly one root $r^*$ over $(1,\infty)$, with $r_*>s_*$, $f(s) > 0$ if $s\in(1,r_*)$ and $f(s)<0$ if $s>r_*$. Since the solution $r_\beta$ of \eqref{eq:r} is an increasing function of $\beta$ and goes to infinity as $\beta$ goes to infinity,  we obtain that, for each $d$, there is a unique $\beta_d>0$ such that $r_\beta=r_*$ (which is simply $\beta = r_*-r_*^{1-d}$). For $\beta <\beta_d$ (resp. $>$), $r_\beta < r_*$ (resp. $>$), hence $f(r_\beta) > 0$ (resp. $<$), so that $c_{sM}>0$ (resp. $<0$).

In the asynchronised case,
\begin{eqnarray*}
c_{sM}^{async} = \ln \po \frac{\hat y_2^\intercal x_1}{\hat y_1^\intercal x_1} \pf  & = &  \ln \po\frac{(r-\beta) d r }{(d-1)(r-\beta) + r }\pf \\
 & = &  \ln \po\frac{ d r^{2-d} }{(d-1)r^{1-d} + r }\pf \ = \ \ln \po\frac{ d r }{d-1 + r^d } \pf \,.
\end{eqnarray*} 
This has the same sign as $f(r)$ with $f(s) = -s^d + ds -(d-1)$. Since  $f'(s)  = d(1-s^{d-1} )$, $f$ is increasing over  $[0,1]$ and decreasing over $[1, + \infty)$, reaching its global maximum at $1$, with $f(1)=0$. In particular, for $\beta\neq 0$, $r\neq 1$ and thus $f(r)< 0$ from which  $c_{sM}^{async} < 0$. 

The periodic case follows from the discussion at the end of Section~\ref{subsec:general}
\subsection*{Acknowledgments}

The research of P. Monmarch\'{e} is supported by the projects SWIDIMS (ANR-20-CE40-0022) and CONVIVIALITY (ANR-23-CE40-0003) of the French National Research Agency. SJS was supported by the U.S. National Science Foundation Grant DEB-2243076.

\bibliographystyle{plain}
\bibliography{biblio}

\begin{thebibliography}{35}
\providecommand{\natexlab}[1]{#1}
\providecommand{\url}[1]{\texttt{#1}}
\expandafter\ifx\csname urlstyle\endcsname\relax
  \providecommand{\doi}[1]{doi: #1}\else
  \providecommand{\doi}{doi: \begingroup \urlstyle{rm}\Url}\fi

\bibitem[Bena{\"\i}m and Strickler(2019)]{benaim2019}
M.~Bena{\"\i}m and E.~Strickler.
\newblock Random switching between vector fields having a common zero.
\newblock \emph{Ann. Appl. Probab.}, 29\penalty0 (1):\penalty0 326--375, 02
  2019.
\newblock \doi{10.1214/18-AAP1418}.
\newblock URL \url{https://doi.org/10.1214/18-AAP1418}.

\bibitem[Bena\"{\i}m et~al.(2014)Bena\"{\i}m, Le~Borgne, Malrieu, and
  Zitt]{Benaim2014Stability}
M.~Bena\"{\i}m, S.~Le~Borgne, F.~Malrieu, and P.-A. Zitt.
\newblock On the stability of planar randomly switched systems.
\newblock \emph{Ann. Appl. Probab.}, 24\penalty0 (1):\penalty0 292--311, 2014.
\newblock ISSN 1050-5164.
\newblock \doi{10.1214/13-AAP924}.
\newblock URL \url{https://doi.org/10.1214/13-AAP924}.

\bibitem[Bena{\"\i}m et~al.(2023)Bena{\"\i}m, Lobry, Sari, and
  Strickler]{BLSS23}
M.~Bena{\"\i}m, C.~Lobry, T.~Sari, and E.~Strickler.
\newblock Untangling the role of temporal and spatial variations in persistence
  of populations.
\newblock \emph{Theoretical Population Biology}, 154:\penalty0 1--26, 2023.

\bibitem[{Bena{\"\i}m} et~al.(2023){Bena{\"\i}m}, {Lobry}, {Sari}, and
  {Strickler}]{NoteTopLyapunov}
M.~{Bena{\"\i}m}, C.~{Lobry}, T.~{Sari}, and {\'E}.~{Strickler}.
\newblock {A note on the top Lyapunov exponent of linear cooperative systems}.
\newblock \emph{arXiv e-prints}, art. arXiv:2302.05874, Feb. 2023.
\newblock \doi{10.48550/arXiv.2302.05874}.

\bibitem[Benaim et~al.(2024)Benaim, Lobry, Sari, and Strickler]{BLSS24}
M.~Benaim, C.~Lobry, T.~Sari, and E.~Strickler.
\newblock When can a population spreading across sink habitats persist?
\newblock \emph{Journal of Mathematical Biology}, 88\penalty0 (2):\penalty0
  1--56, 2024.

\bibitem[Bena{\"\i}m et~al.(2024)Bena{\"\i}m, Lobry, Sari, and
  Strickler]{BLSS24DIGORDID}
M.~Bena{\"\i}m, C.~Lobry, T.~Sari, and E.~Strickler.
\newblock Dispersal-induced growth or decay in a time-periodic environment.
\newblock \emph{arXiv preprint arXiv:2407.07553}, 2024.

\bibitem[Boyce et~al.(2006)Boyce, Haridas, Lee, and the NCEAS Stochastic
  Demography Working~Group]{boyce-etal-06}
M.~S. Boyce, C.~V. Haridas, C.~T. Lee, and the NCEAS Stochastic Demography
  Working~Group.
\newblock Demography in an increasingly variable world.
\newblock \emph{Trends in Ecology and Evolution}, 21:\penalty0 141--148, 2006.

\bibitem[B{\"u}rger(2000)]{burger-00}
R.~B{\"u}rger.
\newblock \emph{The mathematical theory of selection, recombination, and
  mutation}.
\newblock Chichester: Wiley, 2000.

\bibitem[Caswell(2001)]{caswell-01}
H.~Caswell.
\newblock \emph{Matrix Population Models}.
\newblock Sinauer, Sunderland, Massachuesetts, 2001.

\bibitem[{Chitour} et~al.(2021){Chitour}, {Mazanti}, {Monmarch{\'e}}, and
  {Sigalotti}]{Chitour2021}
Y.~{Chitour}, G.~{Mazanti}, P.~{Monmarch{\'e}}, and M.~{Sigalotti}.
\newblock {On the gap between deterministic and probabilistic Lyapunov
  exponents for continuous-time linear systems}.
\newblock \emph{arXiv e-prints}, art. arXiv:2112.07005, Dec. 2021.

\bibitem[{Davis}(1993)]{DavisBook}
M.~{Davis}.
\newblock \emph{Markov Models and Optimization}.
\newblock Monographs on Statistics and Applied Probability. Chapman and Hall,
  1993.

\bibitem[DeAngelis(2018)]{deangelis-18}
D.~L. DeAngelis.
\newblock \emph{Individual-based models and approaches in ecology: populations,
  communities and ecosystems}.
\newblock CRC Press, 2018.

\bibitem[Du et~al.(2021)Du, Hening, Nguyen, and Yin]{DU2021313}
N.~H. Du, A.~Hening, D.~H. Nguyen, and G.~Yin.
\newblock Dynamical systems under random perturbations with fast switching and
  slow diffusion: Hyperbolic equilibria and stable limit cycles.
\newblock \emph{J. Differential Equations}, 293:\penalty0 313--358, 2021.
\newblock ISSN 0022-0396.
\newblock \doi{https://doi.org/10.1016/j.jde.2021.05.032}.
\newblock URL
  \url{https://www.sciencedirect.com/science/article/pii/S0022039621003223}.

\bibitem[Fay et~al.(2020)Fay, Michler, Laesser, Jeanmonod, and
  Schaub]{fay2020can}
R.~Fay, S.~Michler, J.~Laesser, J.~Jeanmonod, and M.~Schaub.
\newblock Can temporal covariation and autocorrelation in demographic rates
  affect population dynamics in a raptor species?
\newblock \emph{Ecology and Evolution}, 10\penalty0 (4):\penalty0 1959--1970,
  2020.

\bibitem[Gorenstein et~al.(2023)Gorenstein, Wainer, Pausata, Prado, Khodri, and
  Dias]{gorenstein202350}
I.~Gorenstein, I.~Wainer, F.~S. Pausata, L.~F. Prado, M.~Khodri, and P.~L.~S.
  Dias.
\newblock A 50-year cycle of sea surface temperature regulates decadal
  precipitation in the tropical and south atlantic region.
\newblock \emph{Communications Earth \& Environment}, 4\penalty0 (1):\penalty0
  427, 2023.

\bibitem[Hanski(1999)]{hanski-99}
I.~Hanski.
\newblock \emph{Metapopulation Ecology}.
\newblock Oxford Series in Ecology and Evolution. Oxford University Press,
  1999.

\bibitem[Hasegawa et~al.(2022)Hasegawa, Katsuta, Muraki, Heimhofer, Ichinnorov,
  Asahi, Ando, Yamamoto, Murayama, Ohta, et~al.]{hasegawa2022decadal}
H.~Hasegawa, N.~Katsuta, Y.~Muraki, U.~Heimhofer, N.~Ichinnorov, H.~Asahi,
  H.~Ando, K.~Yamamoto, M.~Murayama, T.~Ohta, et~al.
\newblock Decadal--centennial-scale solar-linked climate variations and
  millennial-scale internal oscillations during the early cretaceous.
\newblock \emph{Scientific reports}, 12\penalty0 (1):\penalty0 21894, 2022.

\bibitem[Hastings and Gross(2012)]{hastings2012encyclopedia}
A.~Hastings and L.~J. Gross.
\newblock \emph{Encyclopedia of theoretical ecology}.
\newblock Number~4. Univ of California Press, 2012.

\bibitem[Haviv et~al.(1992)Haviv, Ritov, and Rothblum]{HRR92}
M.~Haviv, Y.~Ritov, and U.~G. Rothblum.
\newblock Taylor expansions of eigenvalues of perturbed matrices with
  applications to spectral radii of nonnegative matrices.
\newblock \emph{Linear Algebra Appl.}, 168:\penalty0 159--188, 1992.
\newblock ISSN 0024-3795.
\newblock \doi{10.1016/0024-3795(92)90293-J}.
\newblock URL \url{https://doi.org/10.1016/0024-3795(92)90293-J}.

\bibitem[Hethcote(2000)]{hethcote-00}
H.~W. Hethcote.
\newblock The mathematics of infectious diseases.
\newblock \emph{SIAM Review}, 42\penalty0 (4):\penalty0 599--653, 2000.
\newblock \doi{10.1137/S0036144500371907}.
\newblock URL \url{http://link.aip.org/link/?SIR/42/599/1}.

\bibitem[Hilde et~al.(2020)Hilde, Gamelon, S{\ae}ther, Gaillard, Yoccoz, and
  P{\'e}labon]{hilde2020demographic}
C.~H. Hilde, M.~Gamelon, B.-E. S{\ae}ther, J.-M. Gaillard, N.~G. Yoccoz, and
  C.~P{\'e}labon.
\newblock The demographic buffering hypothesis: evidence and challenges.
\newblock \emph{Trends in Ecology \& Evolution}, 35\penalty0 (6):\penalty0
  523--538, 2020.

\bibitem[Horn and Johnson(2013)]{HJ13}
R.~A. Horn and C.~R. Johnson.
\newblock \emph{Matrix analysis}.
\newblock Cambridge University Press, Cambridge, second edition, 2013.
\newblock ISBN 978-0-521-54823-6.

\bibitem[Katriel(2022)]{K22}
G.~Katriel.
\newblock Dispersal-induced growth in a time-periodic environment.
\newblock \emph{Journal of Mathematical Biology}, 85\penalty0 (3):\penalty0 24,
  2022.

\bibitem[Kon et~al.(2004)Kon, Saito, and Takeuchi]{kon-etal-04}
R.~Kon, Y.~Saito, and Y.~Takeuchi.
\newblock Permanence of single-species stage-structured models.
\newblock \emph{Journal of Mathematical Biology}, 48:\penalty0 515--528, 2004.

\bibitem[Kortessis et~al.(2023{\natexlab{a}})Kortessis, Glass, Gonzalez,
  Ruktanonchai, Simon, Singer, and Holt]{kortessis-etal-2023}
N.~Kortessis, G.~Glass, A.~Gonzalez, N.~Ruktanonchai, M.~Simon, B.~Singer, and
  R.~Holt.
\newblock Neglected consequences of spatio-temporal heterogeneity and
  dispersal: Metapopulations, the inflationary effect, and real-world
  consequences for public health.
\newblock 2023{\natexlab{a}}.

\bibitem[Kortessis et~al.(2023{\natexlab{b}})Kortessis, Glass, Gonzalez,
  Ruktanonchai, Simon, Singer, and Holt]{kortessis2023neglected}
N.~Kortessis, G.~Glass, A.~Gonzalez, N.~W. Ruktanonchai, M.~W. Simon,
  B.~Singer, and R.~D. Holt.
\newblock Neglected consequences of spatio-temporal heterogeneity and
  dispersal: Metapopulations, the inflationary effect, and real-world
  consequences for public health.
\newblock \emph{bioRxiv}, pages 2023--10, 2023{\natexlab{b}}.

\bibitem[Lewontin and Cohen(1969)]{lewontin-cohen-69}
R.~C. Lewontin and D.~Cohen.
\newblock On population growth in a randomly varying environment.
\newblock \emph{Proceedings of the National Academy of Sciences USA},
  62:\penalty0 1056---1060, 1969.

\bibitem[Mitkowski(2008)]{mitkowski-08}
W.~Mitkowski.
\newblock Dynamical properties of metzler systems.
\newblock \emph{Bulletin of the Polish Academy of Sciences: Technical
  Sciences}, 56\penalty0 (4):\penalty0 309--312, 2008.

\bibitem[{Monmarch{\'e}} and {Strickler}(2023)]{MonmarcheStrickler}
P.~{Monmarch{\'e}} and E.~{Strickler}.
\newblock {Asymptotic expansion of the invariant measurefor Markov-modulated
  ODEs at high frequency}.
\newblock \emph{arXiv e-prints}, art. arXiv:2309.16464, Sept. 2023.
\newblock \doi{10.48550/arXiv.2309.16464}.

\bibitem[Paniw et~al.(2018)Paniw, Ozgul, and
  Salguero-G{\'o}mez]{paniw2018interactive}
M.~Paniw, A.~Ozgul, and R.~Salguero-G{\'o}mez.
\newblock Interactive life-history traits predict sensitivity of plants and
  animals to temporal autocorrelation.
\newblock \emph{Ecology letters}, 21\penalty0 (2):\penalty0 275--286, 2018.

\bibitem[Roy et~al.(2005)Roy, Holt, and Barfield]{roy-etal-05}
M.~Roy, R.~Holt, and M.~Barfield.
\newblock Temporal autocorrelation can enhance the persistence and abundance of
  metapopulations comprised of coupled sinks.
\newblock \emph{American Naturalist}, 166:\penalty0 246--261, 2005.

\bibitem[Schreiber(2010)]{prsb-10}
S.~Schreiber.
\newblock Interactive effects of temporal correlations, spatial heterogeneity,
  and dispersal on population persistence.
\newblock \emph{Proceedings of the Royal Society: Biological Sciences},
  277:\penalty0 1907--1914, 2010.

\bibitem[Schreiber(2023)]{schreiber2023partitioning}
S.~Schreiber.
\newblock Partitioning the impacts of spatial-temporal variation in demography
  and dispersal on metapopulation growth rates.
\newblock \emph{bioRxiv}, pages 2023--11, 2023.

\bibitem[Tuljapurkar(1990)]{tuljapurkar-90}
S.~Tuljapurkar.
\newblock \emph{Population Dynamics in Variable Environments}.
\newblock Springer-Verlag, New York, 1990.

\bibitem[Tuljapurkar and Haridas(2006)]{tuljapurkar2006temporal}
S.~Tuljapurkar and C.~Haridas.
\newblock Temporal autocorrelation and stochastic population growth.
\newblock \emph{Ecology Letters}, 9\penalty0 (3):\penalty0 327--337, 2006.

\end{thebibliography}

\end{document}